\documentclass{article}

\usepackage[english]{babel}

\usepackage[letterpaper,top=2cm,bottom=2cm,left=3cm,right=3cm,marginparwidth=1.75cm]{geometry}

\usepackage{bibentry}
\usepackage{appendix}
\usepackage{verbatim} %% for \verbatiminput
\usepackage{graphicx} %% for the graphics environment
\usepackage{epstopdf} %% for converting EPS to PDF
\usepackage{amsmath}
\usepackage{amsfonts}
\usepackage{amssymb}
\usepackage{amsthm}
\usepackage{mathrsfs}
\usepackage{cite}
\usepackage{physics}
\usepackage{braket}
\usepackage{newlfont} %% font commands
\usepackage{subfigure}
\usepackage{url} %% verbatim with URL sensitive line breaks
\usepackage{textcomp} %% for text companion fonts
\usepackage[usenames, dvipsnames]{color} %% coloring
\usepackage{soul} %% hyphenation
\usepackage{geometry}
\usepackage{multicol} %% multiple columns
\usepackage{listings} %% for source codes
\usepackage{units}
\usepackage{afterpage} %% positioning
\usepackage{subfigure}
\usepackage{dirtytalk}
\usepackage{pdfpages}

\usepackage[utf8]{inputenc}
\DeclareUnicodeCharacter{2212}{-}
\usepackage{authblk}
\usepackage{tikz}
\usetikzlibrary{shapes.misc}
\tikzset{cross/.style={cross out, draw=black, minimum size=2*(#1-\pgflinewidth), inner sep=0pt, outer sep=0pt},
cross/.default={3pt}}
\newtheorem{theorem}{Theorem}
\newtheorem{definition}{Definition}

\title{Quantization of the Hamilton Equations of Motion}
\author[1]{Ramon Jose C. Bagunu}
\author[2]{Eric A. Galapon}
\affil[1]{Department of Physical Sciences and Mathematics, University of the Philippines Manila}
\affil[2]{Theoretical Physics Group, National Institute of Physics, University of the Philippines Diliman}

\begin{document}
\maketitle

\begin{abstract}
One of the fundamental problems in quantum mechanics is finding the correct quantum image of a classical observable that would correspond to experimental measurements. We investigate for the appropriate quantization rule that would yield a Hamiltonian that obeys the quantum analogue of Hamilton's equations of motion, which includes differentiation of operators with respect to another operator. To give meaning to this type of differentiation, Born and Jordan established two definitions called the differential quotients of first type and second type. In this paper we modify the definition for the differential quotient of first type and establish its consistency with the differential quotient of second type for different basis operators corresponding to different quantizations. Theorems and differentiation rules including differentiation of operators with negative powers and multiple differentiation were also investigated. We show that the Hamiltonian obtained from Weyl, simplest symmetric, and Born-Jordan quantization all satisfy the required algebra of the quantum equations of motion.
\end{abstract}

\section{Introduction}
Physical quantities that can be experimentally measured are called observables. In classical mechanics, a one-dimensional particle can have observables represented by real valued functions, $f(q,p)$, of the position $q$ and momentum $p$ at any instant of time. Quantum observables, on the other hand, are generally constructed by means of quantization \cite{j2degossonbook, j3cohenbook, j7degosson2016, j8degosson2011, j9degosson2014, j10degosson2013,  magadan, domingo, d9cohen, bj1925a, pablico, flores1, flores2}. It is a linear mapping of the classical observable $f(q,p)$ into a hermitian operator $\textbf{F}(\textbf{q},\textbf{p})$ in the Hilbert space representation of the configuration space of the particle. This is accomplished by promoting the scalar position $q$ and momentum $p$ into the canonically conjugate position and momentum operators $\textbf{q}$ and $\textbf{p}$, respectively,  defined as
\begin{equation}\label{b1}
\textbf{q}\psi(q)=q\psi(q)\quad\text{  and  }\quad\textbf{p}\psi(q)=-i\hbar\frac{\partial}{\partial q}\psi(q).
\end{equation}
The conjugacy of the position and momentum operators mean that they satisfy the commutation relation 
\begin{equation}\label{b2}
[\textbf{q},\textbf{p}]=i\hbar\textbf{1},
\end{equation}
where $\textbf{1}$ is the identity operator. The quantization of $f(q,p)$ is then obtained by replacing $q$ and $p$ with their operator versions, $\textbf{q}$ and $\textbf{p}$. However, the non-commutativity of $\textbf{q}$ and $\textbf{p}$ raises ambiguity because this implies that there is no unique way of associating a given classical observable with a hermitian operator. This is due to the potentially infinitely many possible ways of arranging the products of $\textbf{q}$ and $\textbf{p}$. For instance, the monomial $f(q,p)=p^mq^n$ for $m,n\geq0$ can be quantized as $\textbf{p}^m\textbf{q}^n$, $\textbf{q}^n\textbf{p}^m$, $\textbf{q}\textbf{p}^m\textbf{q}^{n-1}$, etc or a linear sum of these operators.  In general, we can write the quantized operator for the monomial $p^mq^n$ as the linear sum given by \cite{domingo,bender1}
\begin{equation}\label{b3}
\textbf{A}_{m,n}:=\sum_{j=0}^{n}a_j^{(n)}\textbf{q}^j\textbf{p}^m\textbf{q}^{n-j}\quad\quad\quad\text{or}\quad\quad\quad\ \textbf{A}_{m,n}:=\sum_{j=0}^mb_j^{(m)}\textbf{p}^j\textbf{q}^n\textbf{p}^{m-j},
\end{equation}
where the assignment of the coefficients $a_j^{(n)}$ and $b_j^{(m)}$ are dictated by what is known as an ordering rule. The popular ordering rules include the Weyl, simplest symmetric, and the Born-Jordan ordering \cite{domingo, magadan, j2degossonbook, j7degosson2016, j8degosson2011, j9degosson2014, j10degosson2013}. Additionally, the normalization condition is also satisfied by these operators. That is, 
\begin{equation}\label{b4}
    \sum_{j=0}^na_j^{(n)}=1 \quad\quad\text{and}\quad\quad\sum_{j=0}^mb_j^{(m)}=1.
\end{equation}
This condition is imposed so that the observable $p^mq^n$ can be obtained when going back to the classical regime, \textit{i.e.} taken $\textbf{q}$ and $\textbf{p}$ are replaced by their classical counterparts $q$ and $p$. Since the quantum images formed from quantization in \eqref{b3} is a linear sum of several products of $\textbf{q}$ and $\textbf{p}$, it is a necessity to ensure that the sum of their individual weights $a_j^{(n)}$ and $b_j^{(m)}$ should be $1$.
\\
\indent
For the Weyl ordering, which was introduced in \cite{bender2, 19dgelfand}, the classical function $p^mq^n$ is mapped into a hermitian operator called the \textit{Weyl-ordered forms}, namely
\begin{equation}\label{b5}
	\textbf{T}_{m,n}:=\frac{1}{2^n}\sum_{j=0}^n\binom{n}{j}\textbf{q}^j\textbf{p}^m\textbf{q}^{n-j} \quad\text{or}\quad  \textbf{T}_{m,n}:=\frac{1}{2^m}\sum_{j=0}^m\binom{m}{j}\textbf{p}^j\textbf{q}^n\textbf{p}^{m-j}.
\end{equation}
These basis operators have coefficients given by $a_j^{(n)}=\frac{1}{2^n}\binom{n}{j}$ and $b_j^{(m)}=\frac{1}{2^m}\binom{m}{j}$. Each individual term in \eqref{b5} has a different weight, as indicated by the binomial coefficients. Weyl quantization is the inverse of Wigner transform, a method commonly used to determine the corresponding classical observable of a quantum operator in the phase space formulation of quantum mechanics \cite{j7degosson2016}. Also, Weyl quantization has become the standard ordering rule since it is the most symmetrical and it reflects the covariant property of Hamiltonian dynamics with respect to linear canonical transforms \cite{j9degosson2014,j7degosson2016}. It was also shown \cite{degosson2011} that when using Weyl quantization, Schr{\"o}dinger's equation can be derived from Hamilton's equations of motion, and vice versa. This then leads to mathematical equivalence of the two theories. The Weyl ordered forms are also defined for quantizing monomials with negative power. For instance, the Weyl quantization of the monomial $p^{-m}q^n$ is
\begin{equation}\label{b6}
    \textbf{T}_{-m,n}=\frac{1}{2^n}\sum_{j=0}^n\binom{n}{j}\textbf{q}^j\textbf{p}^{-m}\textbf{q}^{n-j}
\end{equation}
and for the monomial $p^mq^{-n}$ given by
\begin{equation}\label{b7}
    \textbf{T}_{m,-n}=\frac{1}{2^m}\sum_{j=0}^m\binom{m}{j}\textbf{p}^j\textbf{q}^{-n}\textbf{p}^{m-j}.
\end{equation}
\\
\indent
The next well-known quantization is the simplest symmetrization rule. In this ordering rule, the factors $\textbf{p}^m\textbf{q}^n$ and $\textbf{q}^n\textbf{p}^m$ are equally satisfactory quantizations \cite{d9cohen}. The basis operators resulting from this quantization rule are given by
\begin{equation}\label{b8}
    \textbf{S}_{m,n}:=\frac{1}{2}\bigg(\textbf{p}^m\textbf{q}^n+\textbf{q}^n\textbf{p}^m\bigg).
\end{equation}
The operators $\textbf{S}_{m,n}$ can be obtained by setting $a_j^{(n)}=\frac{1}{2}(\delta_{j,0}+\delta_{j,n})$ and $b_j^{(m)}=\frac{1}{2}(\delta_{j,0}+\delta_{j,m}$) in equation \eqref{b3}. Similar to Weyl quantization, we can also define quantum images of monomials with negative powers given by
\begin{equation}\label{b9}
    \textbf{S}_{-m,n}=\frac{1}{2}\bigg(\textbf{p}^{-m}\textbf{q}^n+\textbf{q}^n\textbf{p}^{-m}\bigg)
\end{equation}
and
\begin{equation}\label{b10}
    \textbf{S}_{m,-n}=\frac{1}{2}\bigg(\textbf{p}^m\textbf{q}^{-n}+\textbf{q}^{-n}\textbf{p}^m\bigg).
\end{equation}
\\
\indent
The last ordering rule that we will discuss is the Born-Jordan ordering \cite{j8degosson2011, bj1925a}. Here, the quantum image of the monomial $p^mq^n$ is defined such that all of its possible quantizations have equal weight. This averaging rule results to basis operators $\textbf{B}_{m,n}$, namely
\begin{equation}\label{b11}
	\textbf{B}_{m,n}:=\frac{1}{n+1}\sum_{j=0}^n\textbf{q}^j\textbf{p}^m\textbf{q}^{n-j} \quad\text{or}\quad  \textbf{B}_{m,n}:=\frac{1}{m+1}\sum_{j=0}^m\textbf{p}^j\textbf{q}^n\textbf{p}^{m-j}.
\end{equation}
This ordering rule preserves the equivalence of the Schr{\"o}dinger picture and the Heisenberg picture of quantum mechanics \cite{j2degossonbook,j7degosson2016,j9degosson2014,j10degosson2013}, which will be a main area of concern in this paper. For monomials with negative power, the quantization of $p^{-m}q^n$ is \begin{equation}\label{b12}
	\textbf{B}_{-m,n}:=\frac{1}{n+1}\sum_{j=0}^n\textbf{q}^j\textbf{p}^{-m}\textbf{q}^{n-j}
\end{equation}
and for $p^mq^{-n}$ we have
\begin{equation}\label{b13}
 \textbf{B}_{m,-n}:=\frac{1}{m+1}\sum_{j=0}^m\textbf{p}^j\textbf{q}^{-n}\textbf{p}^{m-j}.
\end{equation}
\\
\indent 
At this point one is likely to ask if which of these quantization rules lead to the correct quantum image. Since the resulting operators from these three quantization rules are different, then their respective measurable values will be different as well. This issue can only be resolved by comparing the spectrum of these quantum images and the measured quantities in the laboratory, which, unfortunately, has not been resolved.
\\
\indent
To give more clarity on the quantization problem, let us review the works of de Gosson \cite{j2degossonbook,j9degosson2014} which discussed the equivalence of the Schr{\"o}dinger and Heisenberg pictures of quantum mechanics. It starts in the Schr{\"o}dinger picture, also known as wave mechanics, where the operators are constant (unless they are explicitly time-dependent) and the states $\psi$ evolve in time as
\begin{equation}\label{p1}
    \psi(t)=\textbf{U}(t,t_0)\psi(t_0)
\end{equation}
where
\begin{equation}\label{p2}
    \textbf{U}(t,t_0)=e^{i\textbf{H}_{\mathcal{S}}(t-t_0)/\hbar}
\end{equation}
is a family of unitary operators, and $\textbf{H}_{\mathcal{S}}$ is the Schr{\"o}dinger Hamiltonian. The time evolution of $\psi$ is given by the Schr{\"o}dinger equation
\begin{equation}\label{p3}
    i\hbar\frac{\partial \psi}{\partial t}=\textbf{H}_{\mathcal{S}}\psi.
\end{equation}
\\
\indent
Continuing their discussion in the Heisenberg picture, the state vectors are time-independent operators while the observables are time-dependent operators. Then for an observable $\textbf{A}_{\mathcal{S}}$ in the Schr{\"o}dinger picture, it becomes a time-dependent operator $\textbf{A}_{\mathcal{H}}$ in the Heisenberg picture which satisfies
\begin{equation}\label{p4}
    i\hbar\frac{d\textbf{A}_{\mathcal{H}}}{dt}=i\hbar\frac{\partial\textbf{A}_{\mathcal{H}}}{\partial t}+[\textbf{A}_{\mathcal{H}},\textbf{H}_{\mathcal{H}}]
\end{equation} 
where $\textbf{H}_{\mathcal{H}}$ is the Hamiltonian in the Heisenberg picture. From their derivation, the equivalence of the two representations of quantum mechanics can be then observed as follows. Suppose a ket in the Schr{\"o}dinger picture given by
\begin{equation}\label{p5}
    |\psi_{\mathcal{S}}(t)\rangle=\textbf{U}(t,t_0) |\psi_{\mathcal{S}}(t_0)\rangle
\end{equation}
becomes the constant ket in the Heisenberg picture as
\begin{equation}\label{p6}
    |\psi_{\mathcal{H}}\rangle=\textbf{U}(t,t_0)^{\dagger} |\psi_{\mathcal{S}}(t)\rangle=|\psi_{\mathcal{S}}(t_0)\rangle
\end{equation}
where we removed its time dependence. For a Hamiltonian $\textbf{H}_{\mathcal{S}}$ in the Schr{\"o}dinger representation, it then becomes  
\begin{equation}\label{p7}
    \textbf{H}_{\mathcal{H}}(t)=\textbf{U}(t,t_0)^{\dagger}\textbf{H}_{\mathcal{S}}\textbf{U}(t,t_0)
\end{equation}
in the Heisenberg representation. Now initially at time $t=t_0$, we have 
\begin{equation}\label{p8}
    \textbf{U}(t_0,t_0)=\textbf{1}
\end{equation}
and so 
\begin{equation}\label{p9}
    \textbf{H}_{\mathcal{H}}(t_0)=\textbf{H}_{\mathcal{S}}.
\end{equation}
Since in the Heisenberg picture the energy is constant, the Hamiltonian operator $\textbf{H}_{\mathcal{H}}(t)$ must be a constant of motion. Therefore, $\textbf{H}_{\mathcal{H}}(t)=\textbf{H}_{\mathcal{H}}(t_0)$ or equivalently,
\begin{equation}\label{p10}
    \textbf{H}_{\mathcal{H}}(t)=\textbf{H}_{\mathcal{S}}
\end{equation}
for all $t$. Now, a consequence of equation \eqref{p10} is that both the Hamiltonian operators in the Schr{\"o}dinger and Heisenberg pictures must be quantized using the same quantization rules. Lastly, it was discussed that if Heisenberg's matrix mechanics is correct and is set to be equivalent with Schr{\"o}dinger's wave mechanics, then the Schr{\"o}dinger Hamiltonian should be quantized using the Born-Jordan quantization since the Heisenberg picture of quantum mechanics breaks down if other quantization rules are used.
\\
\indent
Born and Jordan wrote their paper \say{On quantum mechanics} \cite{bj1925a} in an attempt to give light to Heisenberg's matrix mechanics \cite{heisenberg} and give it a firm mathematical basis, which was then followed up in the multi-dimensional case \cite{bj1926b} where Heisenberg is their co-author. Here, Born and Jordan argued that the classical equations of motion should be consistent with its quantum counterpart. That is, the Hamilton equations of classical observables should also apply to their operator versions in quantum mechanics.
\\
\indent
Let us recall the equations of motion in Hamiltonian mechanics given by
\begin{equation}\label{p11}
\dot{q}=\frac{\partial H}{\partial p} \quad\quad\quad\text{  and  }\quad\quad\quad\dot{p}=-\frac{\partial H}{\partial q},
\end{equation}
which govern the trajectory of a particle in the classical phase space. In terms of Poisson brackets, we have
\begin{equation}\label{p12}
\{H,q\}=-\frac{dq}{dt}\quad\quad\quad\text{  and  }\quad\quad\quad\{H,p\}=-\frac{dp}{dt}.
\end{equation}	
Combining \eqref{p11} and \eqref{p12}, we write the equations of motion as
\begin{equation}\label{p13}
\begin{split}
\{H,q\}=-\frac{\partial H}{\partial p}\quad\quad\quad\text{and}\quad\quad\quad\{H,p\}=\frac{\partial H}{\partial q}.
\end{split}
\end{equation}
Born and Jordan then assumed that the operator counterpart of these classical observables follow this analogy, with the Poisson brackets replaced by the commutator divided by $i\hbar$. Therefore, equations of motion of these quantum observables in the Hilbert space are then constructed as
\begin{equation}\label{p14}
\begin{split}
&[\textbf{H},\textbf{q}]=-i\hbar\frac{\partial \textbf{H}}{\partial \textbf{p}}\quad\quad\quad\text{and}\quad\quad\quad[\textbf{H},\textbf{p}]=i\hbar\frac{\partial \textbf{H}}{\partial \textbf{q}}.
\end{split}
\end{equation}
Equation \eqref{p14} now dictates the dynamics of the operators in the quantum regime. However, we can see that one needs to evaluate a partial derivative of an operator with respect to another operator in order to obtain the quantum equations of motion. This is unfamiliar since we are only used to finding the derivative of functions with respect to a variable, or calculating for the derivative of operators with respect to a variable. We now encounter a problem regarding the validity of the quantum equations of motion. 
\\
\indent
The paper is arranged as follows. In the next section, we discuss in detail the Born-Jordan quantization and the motivation for its emergence. The differential quotient of first type is also introduced here. This is followed by a more in-depth discussion on the differential quotient of first type in Section \ref{2.2}, which is then applied to the basis operators of the Weyl, symmetric, and Born-Jordan quantization. In Section \ref{2.3}, the limit definition of the derivative of operators is presented. Also called the differential quotient of second type, we compare this with the first definition of the derivative and discuss some theorems on taking a derivative of an operator with respect to another operator. This is extended to operators of arbitrary ordering. Lastly, we apply everything to formulate the quantum equations of motion, the quantized version of the classical Hamilton's equations of motion in Section \ref{2.4}. The dynamics of the resulting operators are discussed, along with the comparison of the respective quantum equations of motion for Weyl, symmetric, and Born-Jordan basis operators.

%%%%%new section
\section{The Born-Jordan Quantization}\label{2.1}
We have shown that to establish the quantum equations of motion, it is necessary to obtain a derivative of an operator with respect to another operator. However, some confusion arises when one tries to do so. Unlike the more familiar ordinary calculus, the quantities involved are operators that generally do not commute. We also note that the physical meaning of differentiation with respect to an operator is somehow unclear and  unnecessary most of the time.
\\
\indent
To start, let us revisit the paper of Born and Jordan \cite{bj1925a} which tackled symbolic differentiation for non-commuting observables. The first definition, also called the \say{differential quotient of first type}, states that for an operator 
\begin{equation}\label{j1}
\textbf{y}=\prod_{m=1}^s\textbf{y}_{l_m}=\textbf{y}_{l_1}\textbf{y}_{l_2}\textbf{y}_{l_3}...\textbf{y}_{l_s}
\end{equation}
which is a product of non-commuting operators $\textbf{y}_{l}$, the partial derivative with respect to $\textbf{y}_k$ is given by
\begin{equation}\label{j2}
\bigg(\frac{\partial\textbf{y}}{\partial \textbf{y}_k}\bigg)_1=\sum_{r=1}^s\delta_{l_r,k}\prod_{m=r+1}^s\textbf{y}_{l_m}\prod_{m=1}^{r-1}\textbf{y}_{l_m}.
\end{equation}
The subscript of the partial derivative on the left hand side (LHS) indicates that it is being differentiated using the first definition.
The differentiation was discussed by de Gosson in \cite{j2degossonbook,j9degosson2014}, and we now explain it in more detail. Looking at equation \eqref{j2}, the first step is to look at all the factors from left to right. The Kronecker delta indicates that we only pick a factor equal to $\textbf{y}_k$. Then the first product in \eqref{j2} indicates that we write the terms that come after the factor we picked. Then, the second product indicates that we copy the terms on the left of the said factor. To further understand this definition, let us look at some examples. For instance, let's say we have the operator
\begin{equation}\label{j3}
\textbf{A}=\textbf{p}\textbf{p}\textbf{q}\textbf{q}\textbf{q}\textbf{p},
\end{equation}
 where $\textbf{q}$ and $\textbf{p}$ are the position and momentum operators. Taking the partial derivative with respect to the momentum operator gives
\begin{equation}\label{j4}
\bigg(\frac{\partial\textbf{A}}{\partial \textbf{p}}\bigg)_1=\textbf{p}\textbf{q}\textbf{q}\textbf{q}\textbf{p}+\textbf{q}\textbf{q}\textbf{q}\textbf{p}\textbf{p}+\textbf{p}\textbf{p}\textbf{q}\textbf{q}\textbf{q}.
\end{equation}
The operator we are differentiating $\textbf{A}$ with is $\textbf{p}$, so we choose the factors $\textbf{p}$. Looking at $\textbf{A}$ and its leftmost $\textbf{p}$, the order of factors that come after that is $\textbf{p}\textbf{q}\textbf{q}\textbf{q}\textbf{p}$, and no factor is on its left side. Thus, the first term on the right-hand side (RHS) of \eqref{j4} is $\textbf{p}\textbf{q}\textbf{q}\textbf{q}\textbf{p}$. The next $\textbf{p}$ that we pick is on the second factor of $\textbf{A}$. Copying the terms that come after it means we write $\textbf{q}\textbf{q}\textbf{q}\textbf{p}$. Then the term that precedes it is $\textbf{p}$. When we write these two together, we have $\textbf{q}\textbf{q}\textbf{q}\textbf{p}\textbf{p}$. This is the second term in the RHS of \eqref{j4}. The last $\textbf{p}$ remaining is the fifth factor in $\textbf{A}$. There are no factors that come after it and the factors on its left are $\textbf{p}\textbf{p}\textbf{q}\textbf{q}\textbf{q}$. This is where the third term in the RHS of \eqref{j4} comes from. Thus we have the partial derivative of $\textbf{A}$ with respect to $\textbf{p}$ in \eqref{j4}.
\\
\indent
Next, we take the partial derivative of $\textbf{A}$ with respect to $\textbf{q}$. Following the previous procedure, we pick the factors $\textbf{q}$ in the RHS of \eqref{j3}. For the first $\textbf{q}$, the factors that come after it are $\textbf{q}\textbf{q}\textbf{p}$. On its left side we have $\textbf{p}\textbf{p}$. Thus, the first term in the partial derivative is $\textbf{q}\textbf{q}\textbf{p}\textbf{p}\textbf{p}$. For the second $\textbf{q}$ that we pick, we have $\textbf{q}\textbf{p}$ on the succeeding factors and $\textbf{p}\textbf{p}\textbf{q}$ on its left. When we write these together, we have $\textbf{q}\textbf{p}\textbf{p}\textbf{p}\textbf{q}$ for the second term. For the last $\textbf{q}$ that we pick which is the fourth factor of $\textbf{A}$, there is a $\textbf{p}$ on its right then we have $\textbf{p}\textbf{p}\textbf{q}\textbf{q}$ that comes before it. Then we have $\textbf{p}\textbf{p}\textbf{p}\textbf{q}\textbf{q}$ for the last term of the derivative. The partial derivative then is given by
\begin{equation}\label{j5}
\bigg(\frac{\partial\textbf{A}}{\partial \textbf{q}}\bigg)_1= \textbf{q}\textbf{q}\textbf{p}\textbf{p}\textbf{p}+\textbf{q}\textbf{p}\textbf{p}\textbf{p}\textbf{q}+\textbf{p}\textbf{p}\textbf{p}\textbf{q}\textbf{q}.
\end{equation}
\\
\indent
The second definition of the derivative, also called the \say{differential quotient of second type} \cite{bj1925a}, is more familiar. Given the same operator in \eqref{j1}, this definition states that 
\begin{equation}\label{j6}
\bigg(\frac{\partial\textbf{y}}{\partial \textbf{y}_k}\bigg)_2=\lim\limits_{\delta\rightarrow0}\frac{\textbf{y}(\textbf{y}_{1},\textbf{y}_{2},,...,\textbf{y}_{k}+\delta\textbf{1},...,\textbf{y}_{s})-\textbf{y}(\textbf{y}_{1},\textbf{y}_{2},...,\textbf{y}_{k},...,\textbf{y}_{s})}{\delta},
\end{equation}
which resembles the limit definition of partial derivative in ordinary calculus. The subscript of the partial derivative on the LHS indicates that it is being differentiated using the second definition. Then, from the previous operator in \eqref{j3}, we have
\begin{equation}\label{j7a}
\bigg(\frac{\partial\textbf{A}}{\partial \textbf{p}}\bigg)_2=\lim\limits_{\delta\rightarrow0}\frac{(\textbf{p}+\delta\textbf{1})(\textbf{p}+\delta\textbf{1})\textbf{q}\textbf{q}\textbf{q}(\textbf{p}+\delta\textbf{1})-\textbf{p}\textbf{p}\textbf{q}\textbf{q}\textbf{q}\textbf{p}}{\delta},
\end{equation}
to which the numerator can be expanded such that
\begin{equation}\label{j7b
}
\bigg(\frac{\partial\textbf{A}}{\partial \textbf{p}}\bigg)_2=\lim\limits_{\delta\rightarrow0}\frac{\textbf{p}^2\textbf{q}^3\textbf{p}+\textbf{p}^2\textbf{q}^3\delta+2\delta\textbf{p}\textbf{q}^3\textbf{p}+2\delta^2\textbf{p}\textbf{q}^3+\delta^2\textbf{q}^3\textbf{p}+\delta^3\textbf{q}^3-\textbf{p}^2\textbf{q}^3\textbf{p}}{\delta}
\end{equation}
where it simplifies as
\begin{equation}\label{j7}
\bigg(\frac{\partial\textbf{A}}{\partial \textbf{p}}\bigg)_2=2\textbf{p}\textbf{q}^3\textbf{p}+\textbf{p}^2\textbf{q}^3.
\end{equation}
On the other hand, differentiation with respect to the position operator yields
\begin{equation}\label{j8a}
\bigg(\frac{\partial\textbf{A}}{\partial \textbf{q}}\bigg)_2=\lim\limits_{\delta\rightarrow0}\frac{\textbf{p}\textbf{p}(\textbf{q}+\delta\textbf{1})(\textbf{q}+\delta\textbf{1})(\textbf{q}+\delta\textbf{1})\textbf{p}-\textbf{p}\textbf{p}\textbf{q}\textbf{q}\textbf{q}\textbf{p}}{\delta}.
\end{equation}
The numerator in \eqref{j8a} can be expanded to obtain
\begin{equation}\label{j8b}
\bigg(\frac{\partial\textbf{A}}{\partial \textbf{q}}\bigg)_2=\lim\limits_{\delta\rightarrow0}\frac{\textbf{p}^2\textbf{q}^3\textbf{p}+3\delta\textbf{p}^2\textbf{q}^2\textbf{p}+3\delta^2\textbf{p}^2\textbf{q}\textbf{p}+\delta^3\textbf{p}^3-\textbf{p}^2\textbf{q}^3\textbf{p}}{\delta}
\end{equation}
to which the derivative results to
\begin{equation}\label{j8}
\bigg(\frac{\partial\textbf{A}}{\partial \textbf{q}}\bigg)_2=3\textbf{p}^2\textbf{q}^2\textbf{p}.
\end{equation}
\\
\indent
We note that these differentiation rules apply to arbitrary operators and they are not necessarily equal to each other. For the differential quotient of first type, the partial derivative of a product of operators is invariant under cyclic rearrangement of the factors and the rules of ordinary partial differentiation only applies for commuting observables \cite{bj1925a}. The first definition of the derivative was also said to be more natural compared to the differential quotient of second type since it leads to a simple formulation of the variational principle of quantum mechanics \cite{bj1926b}.

An interesting assertion was made in \cite{bj1926b} which says that specifically for a classical Hamiltonian $H(q,p)$, it should be quantized into an operator $\textbf{H}(\textbf{q},\textbf{p})$ such that
\begin{equation}\label{j9}
\bigg(\frac{\partial\textbf{H}}{\partial \textbf{p}}\bigg)_1=\bigg(\frac{\partial\textbf{H}}{\partial \textbf{p}}\bigg)_2\quad\text{    and    }\quad\bigg(\frac{\partial\textbf{H}}{\partial \textbf{q}}\bigg)_1=\bigg(\frac{\partial\textbf{H}}{\partial \textbf{q}}\bigg)_2.
\end{equation}
That is, the first and second definition of the partial derivative should be consistent with the Hamiltonian operator. As we can see in \eqref{j4} and \eqref{j7}, then \eqref{j5} and \eqref{j8}, the example operator $\textbf{A}$ cannot be a legitimate Hamiltonian if we follow this assertion. With these conditions \eqref{j9}, the task now is to find a quantization rule that fits this claim which, as stated in their paper, is the Born-Jordan ordering that was consequently supported in \cite{j2degossonbook,j9degosson2014}.

We now discuss this first by focusing on the classical Hamiltonian function of the form $H(q,p)=p^mq^n$. Its arbitrarily-quantized operator counterpart is given by
\begin{equation}\label{j10}
\textbf{H}=\sum_{j=0}^na_j^{(n)}\textbf{q}^{n-j}\textbf{p}^m\textbf{q}^j=\sum_{j=0}^mb_j^{(m)}\textbf{p}^{m-j}\textbf{q}^n\textbf{p}^j,
\end{equation}
where the coefficients $a_j^{(n)}$ and $b_j^{(m)}$ are still to be identified. We expand the second form of the Hamiltonian in equation \eqref{j10} to get
\begin{equation}\label{j11}
\textbf{H}=a_0^{(n)}\textbf{p}^m\textbf{q}^n+a_1^{(n)}\textbf{q}\textbf{p}^m\textbf{q}^{n-1}+a_2^{(n)}\textbf{q}^2\textbf{p}^m\textbf{q}^{n-2}+...+a_n^{(n)}\textbf{q}^n\textbf{p}^m.
\end{equation}
Let us take the partial derivative of each term in \eqref{j11}. For the first term, we have 
\begin{equation}\label{add1}
    \bigg(\frac{\partial (\textbf{p}^m\textbf{q}^n)}{\partial \textbf{p}}\bigg)_1=\textbf{p}^{m-1}\textbf{q}^n+\textbf{p}^{m-2}\textbf{q}^n\textbf{p}+\textbf{p}^{m-3}\textbf{q}^n\textbf{p}^2+...+\textbf{q}^n\textbf{p}^{m-1}
\end{equation}
which has $m$ terms. Now the derivative of the second term in \eqref{j11} gives 
\begin{equation}\label{add2}
    \bigg(\frac{\partial (\textbf{q}\textbf{p}^m\textbf{q}^{n-1})}{\partial \textbf{p}}\bigg)_1=\textbf{p}^{m-1}\textbf{q}^n+\textbf{p}^{m-2}\textbf{q}^n\textbf{p}+\textbf{p}^{m-3}\textbf{q}^n\textbf{p}^2+...+\textbf{q}^n\textbf{p}^{m-1}
\end{equation}
which is the same as \eqref{add1}. This result can be obtained from all the terms in \eqref{j11}. Therefore, the partial derivative of the Hamiltonian is
\begin{equation}\label{j12}
\begin{split}
\bigg(\frac{\partial\textbf{H}}{\partial \textbf{p}}\bigg)_1&=a_0^{(n)}(\textbf{p}^{m-1}\textbf{q}^n+\textbf{p}^{m-2}\textbf{q}^n\textbf{p}+...+\textbf{q}^n\textbf{p}^{m-1})\\&+a_1^{(n)}(\textbf{p}^{m-1}\textbf{q}^n+\textbf{p}^{m-2}\textbf{q}^n\textbf{p}+...+\textbf{q}^n\textbf{p}^{m-1})\\&+a_2^{(n)}(\textbf{p}^{m-1}\textbf{q}^n+\textbf{p}^{m-2}\textbf{q}^n\textbf{p}+...+\textbf{q}^n\textbf{p}^{m-1})\\&+...\\&+a_n^{(n)}(\textbf{p}^{m-1}\textbf{q}^n+\textbf{p}^{m-2}\textbf{q}^n\textbf{p}+...+\textbf{q}^n\textbf{p}^{m-1}),
\end{split}
\end{equation}
which can be simplified as
\begin{equation}\label{j13}
\bigg(\frac{\partial\textbf{H}}{\partial \textbf{p}}\bigg)_1=\sum_{k=0}^na_k^{(n)}\sum_{j=0}^{m-1}\textbf{p}^j\textbf{q}^n\textbf{p}^{m-1-j}.
\end{equation}
Imposing that $\textbf{H}$ is normalized, we set $\sum_{k=0}^na_k^{(n)}=1$. Therefore,
\begin{equation}\label{j14}
\bigg(\frac{\partial\textbf{H}}{\partial \textbf{p}}\bigg)_1=\sum_{j=0}^{m-1}\textbf{p}^j\textbf{q}^n\textbf{p}^{m-1-j}.
\end{equation}

For the second definition of the derivative which acts like an ordinary partial derivative, we have 
\begin{equation}\label{add3}
    \bigg(\frac{\partial(\textbf{p}^m\textbf{q}^n)}{\partial \textbf{p}}\bigg)_2=m\textbf{p}^{m-1}\textbf{q}^n,
\end{equation}
\begin{equation}\label{add4}
    \bigg(\frac{\partial(\textbf{q}\textbf{p}^m\textbf{q}^{n-1})}{\partial \textbf{p}}\bigg)_2=m\textbf{q}\textbf{p}^{m-1}\textbf{q}^{n-1},
\end{equation}
\begin{equation}\label{add5}
    \bigg(\frac{\partial(\textbf{q}^2\textbf{p}^m\textbf{q}^{n-2})}{\partial \textbf{p}}\bigg)_2=m\textbf{q}^2\textbf{p}^{m-1}\textbf{q}^{n-2}
\end{equation}
and so on. Therefore, applying this partial derivative to the Hamiltonian in \eqref{j11} gives
\begin{equation}\label{j15}
\begin{split}
\bigg(\frac{\partial\textbf{H}}{\partial \textbf{p}}\bigg)_2&=a_0^{(n)}m\textbf{p}^{m-1}\textbf{q}^n+a_1^{(n)}m\textbf{q}\textbf{p}^{m-1}\textbf{q}^{n-1}+a_2^{(n)}m\textbf{q}^2\textbf{p}^{m-1}\textbf{q}^{n-2}...+a_n^{(n)}m\textbf{q}^n\textbf{p}^{m-1}\\&=m\sum_{j=0}^na_j^{(n)}\textbf{q}^j\textbf{p}^{m-1}\textbf{q}^{n-j}.
\end{split}
\end{equation}
Equating the two derivatives from \eqref{j14} and \eqref{j15} yields
\begin{equation}\label{j16}
\sum_{j=0}^{m-1}\textbf{p}^j\textbf{q}^n\textbf{p}^{m-1-j}=m\sum_{j=0}^na_j^{(n)}\textbf{q}^j\textbf{p}^{m-1}\textbf{q}^{n-j}.
\end{equation}
Let us recall the operator in the Born-Jordan basis given by 
\begin{equation}\label{add6}
    \textbf{B}_{m-1,n}=\frac{1}{n+1}\sum_{j=0}^n\textbf{q}^j\textbf{p}^{m-1}\textbf{q}^{n-j}=\frac{1}{m}\sum_{j=0}^{m-1}\textbf{p}^j\textbf{q}^n\textbf{p}^{m-1-j}
\end{equation}
so we can obtain
%\begin{equation}\label{j17}
%\frac{1}{n+1}\sum_{j=0}^n\textbf{q}^j\textbf{p}^{m-1}\textbf{q}^{n-j}=\frac{1}{m}\sum_{j=0}^{m-1}\textbf{p}^j\textbf{q}^n\textbf{p}^{m-1-j}
%\end{equation}
\begin{equation}\label{j18}
\sum_{j=0}^{m-1}\textbf{p}^j\textbf{q}^n\textbf{p}^{m-1-j}=\frac{m}{n+1}\sum_{j=0}^n\textbf{q}^j\textbf{p}^{m-1}\textbf{q}^{n-j}.
\end{equation}
Notice that both the left hand sides of \eqref{j16} and \eqref{j18} are equal. We then equate their right hand sides which leads to
\begin{equation}\label{j19}
m\sum_{j=0}^na_j^{(n)}\textbf{q}^j\textbf{p}^{m-1}\textbf{q}^{n-j}=\frac{m}{n+1}\sum_{j=0}^n\textbf{q}^j\textbf{p}^{m-1}\textbf{q}^{n-j}.
\end{equation}
The unknown coefficients $a_j^{(n)}$ can now be evaluated by inspection to be
\begin{equation}\label{j20}
a_j^{(n)}=\frac{1}{n+1}.
\end{equation}
Therefore, the Hamiltonian operator takes the form
\begin{equation}\label{j21}
\textbf{H}=\frac{1}{n+1}\sum_{j=0}^n\textbf{q}^{n-j}\textbf{p}^m\textbf{q}^j.
\end{equation}
Next to evaluate the coefficients $b_j^{(m)}$, we first expand the second form of the Hamiltonian given by
\begin{equation}\label{j22}
\textbf{H}=b_0^{(m)}\textbf{q}^n\textbf{p}^m+b_1^{(m)}\textbf{p}\textbf{q}^n\textbf{p}^{m-1}+b_2^{(m)}\textbf{p}^2\textbf{q}^n\textbf{p}^{m-2}+...+b_m^{(m)}\textbf{p}^m\textbf{q}^n.
\end{equation}
Using the first definition, we take the partial derivative of each term with respect to $\textbf{q}$. For the first term in \eqref{j22}, we have
\begin{equation}\label{add7}
    \bigg(\frac{\partial (\textbf{q}^n\textbf{p}^m)}{\partial \textbf{q}}\bigg)_1=\textbf{q}^{n-1}\textbf{p}^m+\textbf{1}^{n-2}\textbf{p}^m\textbf{q}+\textbf{q}^{n-3}\textbf{p}^m\textbf{q}^2+...+\textbf{p}^m\textbf{q}^{n-1}
\end{equation}
which has $n$ terms. Now the derivative of the second term in \eqref{j22} gives 
\begin{equation}\label{add8}
    \bigg(\frac{\partial (\textbf{p}\textbf{q}^n\textbf{p}^{m-1})}{\partial \textbf{q}}\bigg)_1=\textbf{q}^{n-1}\textbf{p}^m+\textbf{1}^{n-2}\textbf{p}^m\textbf{q}+\textbf{q}^{n-3}\textbf{p}^m\textbf{q}^2+...+\textbf{p}^m\textbf{q}^{n-1}
\end{equation}
which is the same as \eqref{add7}. This result can be observed from all the terms in \eqref{j22}. Therefore, the partial derivative of the Hamiltonian is
\begin{equation}\label{add9}
\begin{split}
\bigg(\frac{\partial\textbf{H}}{\partial \textbf{q}}\bigg)_1&=b_0^{(m)}(\textbf{q}^{n-1}\textbf{p}^m+\textbf{q}^{n-2}\textbf{p}^m\textbf{q}+...+\textbf{p}^m\textbf{q}^{n-1})\\&+b_1^{(m)}(\textbf{q}^{n-1}\textbf{p}^m+\textbf{q}^{n-2}\textbf{p}^m\textbf{q}+...+\textbf{p}^m\textbf{q}^{n-1})\\&+b_2^{(m)}(\textbf{q}^{n-1}\textbf{p}^m+\textbf{q}^{n-2}\textbf{p}^m\textbf{q}+...+\textbf{p}^m\textbf{q}^{n-1})\\&+...\\&+b_m^{(m)}(\textbf{q}^{n-1}\textbf{p}^m+\textbf{q}^{n-2}\textbf{p}^m\textbf{q}+...+\textbf{p}^m\textbf{q}^{n-1}),
\end{split}
\end{equation}
which reduces to
\begin{equation}\label{j23}
\bigg(\frac{\partial\textbf{H}}{\partial \textbf{q}}\bigg)_1=\sum_{k=0}^mb_k^{(m)}\sum_{j=0}^{n-1}\textbf{q}^j\textbf{p}^m\textbf{q}^{n-1-j}.
\end{equation}
Again, we use the normalization condition that $\sum_{k=0}^mb_k^{(m)}=1$ so we have
\begin{equation}\label{add10}
\bigg(\frac{\partial\textbf{H}}{\partial \textbf{q}}\bigg)_1=\sum_{j=0}^{n-1}\textbf{q}^j\textbf{p}^m\textbf{q}^{n-1-j}.
\end{equation}
The derivative from the second definition then is applied to each term in the Hamiltonian in \eqref{j22} to obtain
\begin{equation}\label{add11}
    \bigg(\frac{\partial(\textbf{q}^{n}\textbf{p}^m)}{\partial \textbf{q}}\bigg)_2=n\textbf{q}^{n-1}\textbf{p}^{m},
\end{equation}
\begin{equation}\label{add12}
    \bigg(\frac{\partial(\textbf{p}\textbf{q}^n\textbf{p}^{m-1})}{\partial \textbf{q}}\bigg)_2=n\textbf{p}\textbf{q}^{n-1}\textbf{p}^{m-1},
\end{equation}
\begin{equation}\label{add13}
    \bigg(\frac{\partial(\textbf{p}^2\textbf{q}^n\textbf{p}^{m-2})}{\partial \textbf{q}}\bigg)_2=n\textbf{p}^2\textbf{q}^{n-1}\textbf{p}^{m-2}
\end{equation}
and so on. Therefore, applying this partial derivative to the Hamiltonian in \eqref{j22} gives
\begin{equation}\label{j24}
\begin{split}
\bigg(\frac{\partial\textbf{H}}{\partial \textbf{q}}\bigg)_2&=b_0^{(m)}n\textbf{q}^{n-1}\textbf{p}^m+b_1^{(m)}n\textbf{q}\textbf{p}^{m-1}\textbf{q}^{n-1}+b_2^{(m)}n\textbf{p}^2\textbf{q}^{n-1}\textbf{p}^{m-2}...+b_m^{(n)}n\textbf{p}^m\textbf{q}^{n-1}\\&=n\sum_{j=0}^mb_j^{(m)}\textbf{p}^j\textbf{q}^{n-1}\textbf{p}^{m-j}.
\end{split}
\end{equation}
Before equating \eqref{add10} and \eqref{j24}, we use first the equality
\begin{equation}\label{j25}
\textbf{B}_{m,n-1}=\frac{1}{n}\sum_{j=0}^{n-1}\textbf{q}^j\textbf{p}^m\textbf{q}^{n-1-j}=\frac{1}{m+1}\sum_{j=0}^m\textbf{p}^j\textbf{q}^{n-1}\textbf{p}^{m-j}
\end{equation}
to obtain
\begin{equation}\label{j26}
\sum_{j=0}^{n-1}\textbf{q}^j\textbf{p}^m\textbf{q}^{n-1-j}=\frac{n}{m+1}\sum_{j=0}^m\textbf{p}^j\textbf{q}^{n-1}\textbf{p}^{m-j}.
\end{equation}
We replace \eqref{add10} by \eqref{j26} then equate it to \eqref{j24} which is given by
\begin{equation}\label{j27}
\frac{n}{m+1}\sum_{j=0}^m\textbf{p}^j\textbf{q}^{n-1}\textbf{p}^{m-j}=n\sum_{j=0}^nb_j^{(m)}\textbf{p}^j\textbf{q}^{n-1}\textbf{p}^{m-j}.
\end{equation}
By inspection, we see that the equation \eqref{j27} is true provided that
\begin{equation}\label{j28}
b_j^{(m)}=\frac{1}{m+1}.
\end{equation}
The Hamiltonian can now be written as
\begin{equation}\label{j29}
\textbf{H}=\frac{1}{m+1}\sum_{j=0}^m\textbf{p}^j\textbf{q}^{n}\textbf{p}^{m-j}
\end{equation}

We can now observe that from the classical Hamiltonian $H(p,q)=p^mq^n$, the operator $\textbf{H}(\textbf{p},\textbf{q})$ that can satisfy the equalities given by condition \eqref{j9} is given by
\begin{equation}\label{j30}
\textbf{H}=\frac{1}{n+1}\sum_{j=0}^n\textbf{q}^{n-j}\textbf{p}^m\textbf{q}^j=\frac{1}{m+1}\sum_{j=0}^m\textbf{p}^j\textbf{q}^{n}\textbf{p}^{m-j}
\end{equation}
which are identical to the operators $\textbf{B}_{m,n}$ from the Born-Jordan quantization. We can now see that for the two definitions of the partial derivative to be consistent, the classical function $p^mq^n$ should be quantized into an operator given by $\textbf{B}_{m,n}$ from the Born-Jordan quantization. This also states that
\begin{equation}\label{j31}
\bigg(\frac{\partial\textbf{B}_{m,n}}{\partial \textbf{p}}\bigg)_1=\bigg(\frac{\partial\textbf{B}_{m,n}}{\partial \textbf{p}}\bigg)_2=m\textbf{B}_{m-1,n}
\end{equation}
and
\begin{equation}\label{j32}
\bigg(\frac{\partial\textbf{B}_{m,n}}{\partial \textbf{q}}\bigg)_1=\bigg(\frac{\partial\textbf{B}_{m,n}}{\partial \textbf{q}}\bigg)_2=n\textbf{B}_{m,n-1}.
\end{equation}
For this reason, Born-Jordan quantization stands as the preferred rule when it comes to the equivalence of the first and second definition of differentiation involving operators.
%% section start
\section{More on the Differential Quotient of First Type}\label{2.2}
Previously, we have shown that the Born-Jordan quantization is necessary to satisfy the imposed conditions in \eqref{j9}. Now, how about the other quantization rules? What differentiation conditions do they satisfy? Here we investigate some properties that the differential quotient of first type shows, when applied to operators that are obtained from Weyl and simplest symmetric ordering.
\\
\indent
Let us first consider a classical Hamiltonian $H(q,p)=p^mq^n$, quantized using Weyl ordering rule to obtain $\textbf{H}(\textbf{q},\textbf{p})=\textbf{T}_{m,n}$. This takes the form
\begin{equation}\label{j33}
\textbf{T}_{m,n}=\frac{1}{2^n}\sum_{j=0}^n\binom{n}{j}\textbf{q}^j\textbf{p}^m\textbf{q}^{n-j}=\frac{1}{2^m}\sum_{j=0}^m\binom{m}{j}\textbf{p}^j\textbf{q}^n\textbf{p}^{m-j}.
\end{equation}
Differentiating each term in the first expression for $\textbf{T}_{m,n}$ gives
\begin{equation}\label{add14}
    \bigg(\frac{\partial (\textbf{p}^m\textbf{q}^n)}{\partial \textbf{p}}\bigg)_1=\textbf{p}^{m-1}\textbf{q}^n+\textbf{p}^{m-2}\textbf{q}^n\textbf{p}+\textbf{p}^{m-3}\textbf{q}^n\textbf{p}^2+...+\textbf{q}^n\textbf{p}^{m-1}
\end{equation}
which has $m$ terms. For the second term we have
\begin{equation}\label{add15}
    \bigg(\frac{\partial (\textbf{q}\textbf{p}^m\textbf{q}^{n-1})}{\partial \textbf{p}}\bigg)_1=\textbf{p}^{m-1}\textbf{q}^n+\textbf{p}^{m-2}\textbf{q}^n\textbf{p}+\textbf{p}^{m-3}\textbf{q}^n\textbf{p}^2+...+\textbf{q}^n\textbf{p}^{m-1}
\end{equation}
which is the same as \eqref{add14}. This result can be obtained for all the terms in $\textbf{T}_{m,n}$. Therefore,
\begin{equation}\label{j34}
\begin{split}
\bigg(\frac{\partial\textbf{T}_{m,n}}{\partial \textbf{p}}\bigg)_1=&\frac{1}{2^n}\bigg\{\binom{n}{0}\bigg[\textbf{p}^{m-1}\textbf{q}^n+\textbf{p}^{m-2}\textbf{q}^n\textbf{p}+...+\textbf{q}^n\textbf{p}^{m-1}\bigg]\\&\;\;\;+\binom{n}{1}\bigg[\textbf{p}^{m-1}\textbf{q}^n+\textbf{p}^{m-2}\textbf{q}^n\textbf{p}+...+\textbf{q}^n\textbf{p}^{m-1}\bigg]\\&\;\;\;+\binom{n}{2}\bigg[\textbf{p}^{m-1}\textbf{q}^n+\textbf{p}^{m-2}\textbf{q}^n\textbf{p}+...+\textbf{q}^n\textbf{p}^{m-1}\bigg]\\&\;\;\;+...\\&\;\;\;+\binom{n}{n}\bigg[\textbf{p}^{m-1}\textbf{q}^n+\textbf{p}^{m-2}\textbf{q}^n\textbf{p}+...+\textbf{q}^n\textbf{p}^{m-1}\bigg]\bigg\}
\end{split}
\end{equation}
which becomes
\begin{equation}\label{j35}
\bigg(\frac{\partial\textbf{T}_{m,n}}{\partial \textbf{p}}\bigg)_1=\frac{1}{2^n}\sum_{k=0}^n\binom{n}{k}\sum_{j=0}^{m-1}\textbf{p}^j\textbf{q}^n\textbf{p}^{m-1-j}.
\end{equation}
Let us recall that $\sum_{k=0}^n\binom{n}{k}=2^n$. Therefore,
\begin{equation}\label{add16}
\bigg(\frac{\partial\textbf{T}_{m,n}}{\partial \textbf{p}}\bigg)_1=\sum_{j=0}^{m-1}\textbf{p}^j\textbf{q}^n\textbf{p}^{m-1-j}.
\end{equation}
or equivalently
\begin{equation}\label{j36}
\bigg(\frac{\partial\textbf{T}_{m,n}}{\partial \textbf{p}}\bigg)_1=m\textbf{B}_{m-1,n}.
\end{equation}
We can notice that even though we started with the Weyl-ordered operator $\textbf{T}_{m,n}$, the derivative takes the form of the one in the Born-Jordan quantization. The same result can also be observed when taking the partial derivative with respect to $\textbf{q}$. That is,
\begin{equation}\label{j37}
\bigg(\frac{\partial\textbf{T}_{m,n}}{\partial \textbf{q}}\bigg)_1=n\textbf{B}_{m,n-1}.
\end{equation}
\\
We note that the first definition derivatives in \eqref{j36} and \eqref{j37} are not equal to their second definition counterparts (see \eqref{d1} and \eqref{e1}).

Now, suppose we start from the operator $\textbf{S}_{m,n}$ in the simplest symmetrization rule and take its derivative. Let us recall that
\begin{equation}\label{j38}
\textbf{S}_{m,n}=\frac{1}{2}(\textbf{p}^m\textbf{q}^n+\textbf{q}^n\textbf{p}^m).
\end{equation}
Taking the partial derivative with respect to the momentum gives
\begin{equation}\label{j39}
\begin{split}
\bigg(\frac{\partial\textbf{S}_{m,n}}{\partial \textbf{p}}\bigg)_1=&\frac{1}{2}\bigg\{\bigg[\textbf{p}^{m-1}\textbf{q}^n+\textbf{p}^{m-2}\textbf{q}^n\textbf{p}+...+\textbf{q}^n\textbf{p}^{m-1}\bigg]\\&\;+\bigg[\textbf{p}^{m-1}\textbf{q}^n+\textbf{p}^{m-2}\textbf{q}^n\textbf{p}+...+\textbf{q}^n\textbf{p}^{m-1}\bigg]\bigg\}
\end{split}
\end{equation}
that simplifies to
\begin{equation}\label{j40}
\bigg(\frac{\partial\textbf{S}_{m,n}}{\partial \textbf{p}}\bigg)_1=\sum_{j=0}^{m-1}\textbf{p}^j\textbf{q}^n\textbf{p}^{m-1-j},
\end{equation}
which can be written as 
\begin{equation}\label{j41}
\bigg(\frac{\partial\textbf{S}_{m,n}}{\partial \textbf{p}}\bigg)_1=m\textbf{B}_{m-1,n}.
\end{equation}
Similarly, the derivative with respect to $\textbf{q}$ is given by
\begin{equation}\label{j42}
\begin{split}
\bigg(\frac{\partial\textbf{S}_{m,n}}{\partial \textbf{q}}\bigg)_1=&\frac{1}{2}\bigg\{\bigg[\textbf{q}^{n-1}\textbf{p}^m+\textbf{q}^{n-2}\textbf{p}^m\textbf{q}+...+\textbf{p}^m\textbf{q}^{n-1}\bigg]\\&\;+\bigg[\textbf{q}^{n-1}\textbf{p}^m+\textbf{q}^{n-2}\textbf{p}^m\textbf{q}+...+\textbf{p}^m\textbf{q}^{n-1}\bigg]\bigg\}
\end{split}
\end{equation}
that can be written as
\begin{equation}\label{j43}
\bigg(\frac{\partial\textbf{S}_{m,n}}{\partial \textbf{q}}\bigg)_1=\sum_{j=0}^{n-1}\textbf{q}^j\textbf{p}^m\textbf{q}^{n-1-j},
\end{equation}
reducing to
\begin{equation}\label{j44}
\bigg(\frac{\partial\textbf{S}_{m,n}}{\partial \textbf{q}}\bigg)_1=n\textbf{B}_{m,n-1}.
\end{equation}
Again, we note that the first definition derivatives in \eqref{j41} and \eqref{j44} are not equal to their second definition counterparts (see \eqref{d1} and \eqref{e1}).

An observation can now be drawn from the obtained derivatives. Notice that
\begin{equation}\label{j45}
\bigg(\frac{\partial\textbf{B}_{m,n}}{\partial \textbf{p}}\bigg)_1=\bigg(\frac{\partial\textbf{T}_{m,n}}{\partial \textbf{p}}\bigg)_1=\bigg(\frac{\partial\textbf{S}_{m,n}}{\partial \textbf{p}}\bigg)_1=m\textbf{B}_{m-1,n}
\end{equation}
and
\begin{equation}\label{j46}
\bigg(\frac{\partial\textbf{B}_{m,n}}{\partial \textbf{q}}\bigg)_1=\bigg(\frac{\partial\textbf{T}_{m,n}}{\partial \textbf{q}}\bigg)_1=\bigg(\frac{\partial\textbf{S}_{m,n}}{\partial \textbf{q}}\bigg)_1=n\textbf{B}_{m,n-1}
\end{equation}
\\
Whatever basis we are starting with, the end result will always lean toward the Born-Jordan quantized operators. Somehow, the differential quotient of first type has a way of tailoring the derivative in favor of Born-Jordan quantization. We now arrive at the remark that the differential quotient of first type is consistent only when we are working with the basis operators in Born-Jordan quantization. This rule does not apply on operators $\textbf{T}_{m,n}$ and $\textbf{S}_{m,n}$. Consequently, this also raises an issue on the assertion made in \cite{bj1925a} given by the conditions in \eqref{j9}. Another fact that we consider is that the differential quotient of first type changes its value depending whether operators are written in normal order (momentum operators to the right) or antinormal order (momentum operators to the left). Since there is a lack of physical motivation for this differentiation rule, in some way we can also establish a definition that works on other basis operators and not limited only to operators $\textbf{B}_{m,n}$.
\\
\indent
Let us now consider the following case. Suppose we define a derivative such that
\begin{equation}\label{j47}
\bigg(\frac{\partial \textbf{y}}{\partial \textbf{y}_k}\bigg)_{W}=\frac{n}{2^{n-1}}\sum_{r=1}^s\delta_{l_r,k}\binom{n-1}{l-1}\prod_{m=r+1}^s\textbf{y}_{l_m}\prod_{m=1}^{r-1}\textbf{y}_{l_m}
\end{equation}
where $n$ is the number of $\textbf{y}_k$ terms written as
\begin{equation}\label{j48}
n=\sum_{r=1}^s\delta_{l_r,k}
\end{equation}
and $l$ is the number of $\textbf{y}_k$ terms up to the $r$th term given by
\begin{equation}\label{j49}
l=\sum_{t=1}^r\delta_{l_t,k}.
\end{equation}
We now use this differentiation rule \eqref{j47} and apply it to different operators. Starting with the Weyl operator $\textbf{T}_{m,n}$, its partial derivative with respect to $\textbf{p}$ is
\begin{equation}\label{j50}
\begin{split}
\bigg(\frac{\partial\textbf{T}_{m,n}}{\partial \textbf{p}}\bigg)_W=&\frac{m}{2^{m-1}}\frac{1}{2^n}\bigg\{\binom{n}{0}\bigg[\binom{m-1}{0}\textbf{p}^{m-1}\textbf{q}^n+\binom{m-1}{1}\textbf{p}^{m-2}\textbf{q}^n\textbf{p}+\\&\quad\quad\quad\quad\quad\quad\quad\quad\quad\quad\quad\quad\quad\;\;...+\binom{m-1}{m-1}\textbf{q}^n\textbf{p}^{m-1}\bigg]\\&\quad\quad\quad+\binom{n}{1}\bigg[\binom{m-1}{0}\textbf{p}^{m-1}\textbf{q}^n+\binom{m-1}{1}\textbf{p}^{m-2}\textbf{q}^n\textbf{p}+\\&\quad\quad\quad\quad\quad\quad\quad\quad\quad\quad\quad\quad\quad\;\;...+\binom{m-1}{m-1}\textbf{q}^n\textbf{p}^{m-1}\bigg]\\&\quad\quad\quad+...\\&\quad\quad\quad+\binom{n}{n}\bigg[\binom{m-1}{0}\textbf{p}^{m-1}\textbf{q}^n+\binom{m-1}{1}\textbf{p}^{m-2}\textbf{q}^n\textbf{p}+\\&\quad\quad\quad\quad\quad\quad\quad\quad\quad\quad\quad\quad\quad\;\;...+\binom{m-1}{m-1}\textbf{q}^n\textbf{p}^{m-1}\bigg]\bigg\},
\end{split}
\end{equation}
that can be simplified to
\begin{equation}\label{j51}
\bigg(\frac{\partial\textbf{T}_{m,n}}{\partial \textbf{p}}\bigg)_W=\frac{m}{2^{m-1}}\sum_{j=0}^{m-1}\binom{m-1}{j}\textbf{p}^j\textbf{q}^n\textbf{p}^{m-1-j},
\end{equation}
which takes the form 
\begin{equation}\label{j52}
\bigg(\frac{\partial\textbf{T}_{m,n}}{\partial \textbf{p}}\bigg)_W=m\textbf{T}_{m-1,n}.
\end{equation}
We also have the derivative with respect to $\textbf{q}$ given by
\begin{equation}\label{j53}
\begin{split}
\bigg(\frac{\partial\textbf{T}_{m,n}}{\partial \textbf{q}}\bigg)_W=&\frac{n}{2^{n-1}}\frac{1}{2^m}\bigg\{\binom{m}{0}\bigg[\binom{n-1}{0}\textbf{q}^{n-1}\textbf{p}^m+\binom{n-1}{1}\textbf{q}^{n-2}\textbf{p}^m\textbf{q}+\\&\quad\quad\quad\quad\quad\quad\quad\quad\quad\quad\quad\quad\quad\;\;...+\binom{n-1}{n-1}\textbf{p}^m\textbf{q}^{n-1}\bigg]\\&\quad\quad\quad+\binom{m}{1}\bigg[\binom{n-1}{0}\textbf{q}^{n-1}\textbf{p}^m+\binom{n-1}{1}\textbf{q}^{n-2}\textbf{p}^m\textbf{q}+\\&\quad\quad\quad\quad\quad\quad\quad\quad\quad\quad\quad\quad\quad\;\;...+\binom{n-1}{n-1}\textbf{p}^m\textbf{q}^{n-1}\bigg]\\&\quad\quad\quad+...\\&\quad\quad\quad+\binom{m}{m}\bigg[\binom{n-1}{0}\textbf{q}^{n-1}\textbf{p}^m+\binom{n-1}{1}\textbf{q}^{n-2}\textbf{p}^m\textbf{q}+\\&\quad\quad\quad\quad\quad\quad\quad\quad\quad\quad\quad\quad\quad\;\;...+\binom{n-1}{n-1}\textbf{p}^m\textbf{q}^{n-1}\bigg]\bigg\}.
\end{split}
\end{equation}
Equation \eqref{j53} can be further simplified as
\begin{equation}\label{j54}
\bigg(\frac{\partial\textbf{T}_{m,n}}{\partial \textbf{q}}\bigg)_W=\frac{n}{2^{n-1}}\sum_{j=0}^{n-1}\binom{n-1}{j}\textbf{q}^j\textbf{p}^m\textbf{q}^{n-1-j},
\end{equation}
which can be written in terms of Weyl-ordered forms given by 
\begin{equation}\label{j55}
\bigg(\frac{\partial\textbf{T}_{m,n}}{\partial \textbf{q}}\bigg)_W=n\textbf{T}_{m,n-1}.
\end{equation}

Now let us apply this modified definition of derivative \eqref{j47} to the operator $\textbf{S}_{m,n}$. Doing so yields
\begin{equation}\label{j56}
\begin{split}
\bigg(\frac{\partial\textbf{S}_{m,n}}{\partial \textbf{p}}\bigg)_W=&\frac{m}{2^{m-1}}\frac{1}{2}\bigg\{\bigg[\binom{m-1}{0}\textbf{p}^{m-1}\textbf{q}^n+\binom{m-1}{1}\textbf{p}^{m-2}\textbf{q}^n\textbf{p}+\\&\quad\quad\quad\quad\quad\quad\quad\quad\quad\quad\quad...+\binom{m-1}{m-1}\textbf{q}^n\textbf{p}^{m-1}\bigg]\\&\quad\quad\quad+\bigg[\binom{m-1}{0}\textbf{p}^{m-1}\textbf{q}^n+\binom{m-1}{1}\textbf{p}^{m-2}\textbf{q}^n\textbf{p}+\\&\quad\quad\quad\quad\quad\quad\quad\quad\quad\quad\quad\;\;...+\binom{m-1}{m-1}\textbf{q}^n\textbf{p}^{m-1}\bigg]\bigg\}.
\end{split}
\end{equation}
Further simplification results in
\begin{equation}\label{j57}
\bigg(\frac{\partial\textbf{S}_{m,n}}{\partial \textbf{q}}\bigg)_W=\frac{m}{2^{m-1}}\sum_{j=0}^{m-1}\binom{m-1}{j}\textbf{p}^j\textbf{p}^n\textbf{p}^{m-1-j},
\end{equation}
which can again be written as a Weyl-ordered form. We have
\begin{equation}\label{j58}
\bigg(\frac{\partial\textbf{S}_{m,n}}{\partial \textbf{p}}\bigg)_W=m\textbf{T}_{m-1,n}.
\end{equation}
We now see a trend here that the modification in the derivative results in an operator that is similar to Weyl-ordered forms, even though we are doing it to basis operators corresponding to different quantizations. Performing this to all operators yields
\begin{equation}\label{j59}
\bigg(\frac{\partial\textbf{T}_{m,n}}{\partial \textbf{p}}\bigg)_W=\bigg(\frac{\partial\textbf{S}_{m,n}}{\partial \textbf{p}}\bigg)_W=\bigg(\frac{\partial\textbf{B}_{m,n}}{\partial \textbf{p}}\bigg)_W=m\textbf{T}_{m-1,n}
\end{equation}
and
\begin{equation}\label{j60}
\bigg(\frac{\partial\textbf{T}_{m,n}}{\partial \textbf{q}}\bigg)_W=\bigg(\frac{\partial\textbf{S}_{m,n}}{\partial \textbf{q}}\bigg)_W=\bigg(\frac{\partial\textbf{B}_{m,n}}{\partial \textbf{q}}\bigg)_W=n\textbf{T}_{m,n-1}.
\end{equation}
\\
\indent
Lastly, let us modify the derivative to have results that lean toward symmetric quantization. If we consider the differentiation rule given by
\begin{equation}\label{j61}
\bigg(\frac{\partial \textbf{y}}{\partial \textbf{y}_k}\bigg)_{S}=\frac{n}{2}\sum_{r=1}^s\delta_{l_r,k}(\delta_{l,1}+\delta_{l,m})\prod_{m=r+1}^s\textbf{y}_{l_m}\prod_{m=1}^{r-1}\textbf{y}_{l_m},
\end{equation}
the partial derivatives are given by
\begin{equation}\label{j62}
\begin{split}
\bigg(\frac{\partial\textbf{S}_{m,n}}{\partial \textbf{p}}\bigg)_S=&\frac{1}{2}\bigg\{\frac{1}{2}\bigg[\textbf{p}^{m-1}\textbf{q}^n+\textbf{q}^n\textbf{p}^{m-1}\bigg]+\frac{1}{2}\bigg[\textbf{p}^{m-1}\textbf{q}^n+\textbf{q}^n\textbf{p}^{m-1}\bigg]\bigg\}
\end{split}
\end{equation}
and 
\begin{equation}\label{j63}
\begin{split}
\bigg(\frac{\partial\textbf{S}_{m,n}}{\partial \textbf{q}}\bigg)_S=&\frac{1}{2}\bigg\{\frac{1}{2}\bigg[\textbf{q}^{n-1}\textbf{p}^m+\textbf{p}^m\textbf{q}^{n-1}\bigg]+\frac{1}{2}\bigg[\textbf{q}^{n-1}\textbf{p}^m+\textbf{p}^m\textbf{q}^{n-1}\bigg]\bigg\}.
\end{split}
\end{equation}
Equations \eqref{j62} and \eqref{j63} can be simplified as
\begin{equation}\label{j64}
\bigg(\frac{\partial\textbf{S}_{m,n}}{\partial \textbf{p}}\bigg)_S=m\textbf{S}_{m-1,n}
\end{equation}
and
\begin{equation}\label{j65}
\bigg(\frac{\partial\textbf{S}_{m,n}}{\partial \textbf{q}}\bigg)_S=n\textbf{S}_{m,n-1}.
\end{equation}
Doing so with other operators yield the following results
\begin{equation}\label{j66}
\bigg(\frac{\partial\textbf{T}_{m,n}}{\partial \textbf{p}}\bigg)_S=\bigg(\frac{\partial\textbf{S}_{m,n}}{\partial \textbf{p}}\bigg)_S=\bigg(\frac{\partial\textbf{B}_{m,n}}{\partial \textbf{p}}\bigg)_S=m\textbf{S}_{m-1,n}
\end{equation}
and
\begin{equation}\label{j67}
\bigg(\frac{\partial\textbf{T}_{m,n}}{\partial \textbf{q}}\bigg)_S=\bigg(\frac{\partial\textbf{S}_{m,n}}{\partial \textbf{q}}\bigg)_S=\bigg(\frac{\partial\textbf{B}_{m,n}}{\partial \textbf{q}}\bigg)_S=n\textbf{S}_{m,n-1}.
\end{equation}
Again, as we have shown before, redefining the original differential quotient of first type to \eqref{j61} leads to a preferred quantization. In this case, all resulting derivatives take the form of the basis operators in simplest symmetric quantization.
\\
\indent
When we compare the different definitions of the derivative, we can see that each modification has a resemblance to different basis operators. For instance in \eqref{j47}, we can see the factor $\frac{n}{2^{n-1}}$,taking into account the normalization factor along with the binomial coefficients $\binom{n-1}{l-1}$. This then leads to operators that resemble the Weyl-ordered forms. In the modification \eqref{j61}, the factors $\frac{n}{2}$, along with the Kronecker delta terms indicate that the resulting derivative will follow the basis operators in symmetric ordering. We have the normalizing factor $\frac{1}{2}$ and the Kronecker delta ensures the non-zero terms that are symmetrically ordered. We note that these modified differentiation rules cannot be generalized into arbitrary ordering, and that they are tailored to be consistent only in their respective basis. %Now in the original definition \eqref{j2}, we do not see any additional factor involved. We can think of it as follows. Suppose a factor of $\frac{n}{n}$ is included, which takes into account the normalization. Then multiply each term by $1$, which are the factors in Born-Jordan quantization. Since coincidentally the resulting additional factors are all equal to $1$, then the original differential quotient of first type needs no modification for it to be tailored such that the resulting derivatives are the basis operators in Born-Jordan quantization. 
%% section end
\section{The Limit Definition of the Derivative of Operators}\label{2.3}
After giving more clarity on the differential quotient of first type, we now discuss the differential quotient of second type, or the usual limit definition of the derivative. We have observed in Section \ref{2.2} that the first definition of derivative and its modifications cannot be used in arbitrarily-ordered operators. Here we investigate for the generality of the second type of derivative, along with some of its properties by presenting some theorems that involve this differentiation rule.

To start, let us state the limit definition of the derivative for classical observables.
%%\begin{definition}
Given a function $f=f(q,p)$, its partial derivative is defined as
\begin{equation}\label{fa1}
\frac{\partial f(q,p)}{\partial p}=\lim\limits_{\delta\rightarrow 0}\frac{f(q,p+\delta)-f(q,p)}{\delta}
\end{equation}
\textit{and}
\begin{equation}\label{fa2}
\frac{\partial f(q,p)}{\partial q}=\lim\limits_{\delta\rightarrow 0}\frac{f(q+\delta,p)-f(q,p)}{\delta}.
\end{equation}
%%\end{definition}
Extending this to operators, we now present the limit definition of the differential quotient.
\begin{definition}
    For an operator $\textbf{F}=\textbf{F}(\textbf{q},\textbf{p})$, its partial derivative is defined as
    \begin{equation}\label{f1}
     \frac{\partial \textbf{F}(\textbf{q},\textbf{p})}{\partial \textbf{q}}=\lim\limits_{\delta\rightarrow 0}\frac{\textbf{F}(\textbf{q}+\delta\textbf{1},\textbf{p})-\textbf{F}(\textbf{q},\textbf{p})}{\delta}    
    \end{equation}
    and 
    \begin{equation}\label{f2}
     \frac{\partial \textbf{F}(\textbf{q},\textbf{p})}{\partial \textbf{p}}=\lim\limits_{\delta\rightarrow 0}\frac{\textbf{F}(\textbf{q},\textbf{p}+\delta\textbf{1})-\textbf{F}(\textbf{q},\textbf{p})}{\delta},  
    \end{equation}
    where $\textbf{1}$ is the identity operator.
\end{definition}
A well-known property of classical differentiation is linearity. We extend this property to operators in the following theorem.
\begin{theorem}\label{theorem0}
    The differential quotient of second type is linear.
    \begin{proof}
        Let $\textbf{F}=\textbf{F}(\textbf{q},\textbf{p})$ and $\textbf{G}=\textbf{G}(\textbf{q},\textbf{p})$. For scalars $a$ and $b$,
\begin{equation}\label{e01}
    \frac{\partial (a\textbf{F}+b\textbf{G})}{\partial \textbf{q}}=\lim\limits_{\delta\rightarrow 0}\frac{(a\textbf{F}+b\textbf{G})(\textbf{q}+\delta\textbf{1},\textbf{p})-(a\textbf{F}+b\textbf{G})(\textbf{q},\textbf{p})}{\delta}
\end{equation}
This can be rewritten as
\begin{equation}\label{e02}
    \frac{\partial (a\textbf{F}+b\textbf{G})}{\partial \textbf{q}}=\lim\limits_{\delta\rightarrow 0}\frac{a\textbf{F}(\textbf{q}+\delta\textbf{1},\textbf{p})+b\textbf{G}(\textbf{q}+\delta\textbf{1},\textbf{p})-a\textbf{F}(\textbf{q},\textbf{p})-b\textbf{G}(\textbf{q},\textbf{p})}{\delta}.
\end{equation}
Rearranging the terms and factoring out the constants, we obtain
\begin{equation}\label{e03}
    \frac{\partial (a\textbf{F}+b\textbf{G})}{\partial \textbf{q}}=a\lim\limits_{\delta\rightarrow 0}\frac{\textbf{F}(\textbf{q}+\delta\textbf{1},\textbf{p})-\textbf{F}(\textbf{q},\textbf{p})}{\delta} + b\lim\limits_{\delta\rightarrow 0}\frac{\textbf{G}(\textbf{q}+\delta\textbf{1},\textbf{p})-\textbf{G}(\textbf{q},\textbf{p})}{\delta}
\end{equation}
which simplifies to 
\begin{equation}\label{e04}
    \frac{\partial (a\textbf{F}+b\textbf{G})}{\partial \textbf{q}}=a\frac{\partial\textbf{F}}{\partial \textbf{q}}+ b\frac{\partial\textbf{G}}{\partial \textbf{q}}.
\end{equation}
Similarly,
\begin{equation}\label{e05}
    \frac{\partial (a\textbf{F}+b\textbf{G})}{\partial \textbf{p}}=a\frac{\partial\textbf{F}}{\partial \textbf{p}}+ b\frac{\partial\textbf{G}}{\partial \textbf{p}}.
\end{equation}
    \end{proof}
\end{theorem}

After establishing the linearity of the differential quotient, we present the following theorems that can be generally applied for arbitrarily-quantized observables.
\begin{theorem}\label{theorem1}
Let $f(q,p)=p^mq^n$ with its arbitrarily-quantized quantum image $\textbf{F}(\textbf{q},\textbf{p})=\textbf{A}_{m,n}$. The partial derivative is given by
\begin{equation}\label{d1}
\frac{\partial f(q,p)}{\partial p}=mp^{m-1}q^n\rightarrow\frac{\partial \textbf{F}(\textbf{q},\textbf{p})}{\partial \textbf{p}}=m\textbf{A}_{m-1,n}
\end{equation}
\textit{and}
\begin{equation}\label{e1}
\frac{\partial f(q,p)}{\partial q}=np^mq^{n-1}\rightarrow\frac{\partial \textbf{F}(\textbf{q},\textbf{p})}{\partial \textbf{q}}=n\textbf{A}_{m,n-1}.
\end{equation}
\end{theorem}
\begin{proof}
We start with the operators
\begin{equation}\label{d2}
\textbf{F}(\textbf{q},\textbf{p})=\textbf{A}_{m,n}=\sum_{j=0}^na^{(n)}_j\textbf{q}^j\textbf{p}^m\textbf{q}^{n-j}
\end{equation}
where $a_j^{(n)}$ depends on the ordering rule.
Extending the definition of derivative to the operators yield
\begin{equation}\label{d4}
\begin{split}
\frac{\partial \textbf{F}(\textbf{q},\textbf{p})}{\partial \textbf{p}}=\lim\limits_{\delta\rightarrow 0}\sum_{j=0}^na_j^{(n)}\frac{\big[\textbf{q}^j(\textbf{p}+\delta\textbf{1})^m\textbf{q}^{n-j}-\textbf{q}^j\textbf{p}^m\textbf{q}^{n-j}\big]}{\delta},
\end{split}
\end{equation}
where we used the linearity of differentiation to factor out the constant $a_j^{(n)}$. Upon expanding the factor $(\textbf{p}+\delta \textbf{1})^m$ in the numerator, we obtain
\begin{equation}\label{d5}
\frac{\partial \textbf{F}(\textbf{q},\textbf{p})}{\partial \textbf{p}}=\lim\limits_{\delta\rightarrow 0}\sum_{j=0}^na_j^{(n)}\frac{\textbf{q}^j\big[\sum_{k=0}^m\binom{m}{k}\textbf{p}^{m-k}\delta^k\big]\textbf{q}^{n-j}-\textbf{q}^j\textbf{p}^m\textbf{q}^{n-j}}{\delta},
\end{equation}
in which the non-vanishing terms are for all values of $j$ and for $k=0, 1$. This gives us
\begin{equation}\label{d6}
\frac{\partial \textbf{F}(\textbf{q},\textbf{p})}{\partial \textbf{p}}=m\sum_{j=0}^na_j^{(n)}\textbf{q}^j\textbf{p}^{m-1}\textbf{q}^{n-j},
\end{equation}
which can be written as
\begin{equation}\label{d7}
\frac{\partial \textbf{F}(\textbf{q},\textbf{p})}{\partial \textbf{p}}=m\textbf{A}_{m-1,n},
\end{equation}
Next, we use another form of the operator
\begin{equation}\label{e2}
\textbf{F}(\textbf{q},\textbf{p})=\textbf{A}_{m,n}=\sum_{j=0}^mb_j^{(m)}\textbf{p}^j\textbf{q}^m\textbf{p}^{m-j}
\end{equation}
where $b_j^{(m)}$ depends on the ordering rule being used, to obtain the partial derivative with respect to $\textbf{q}$. Using the limit definition of the derivative,
we have
\begin{equation}\label{e4}
\frac{\partial \textbf{F}(\textbf{q},\textbf{p})}{\partial \textbf{q}}=\lim\limits_{\delta\rightarrow 0}\sum_{j=0}^mb_j^{(m)}\frac{\big[\textbf{p}^j(\textbf{q}+\delta\textbf{1})^n\textbf{p}^{m-j}-\textbf{p}^j\textbf{q}^n\textbf{p}^{m-j}\big]}{\delta}
\end{equation}
in which the term $(\textbf{q}+\delta\textbf{1})^n$ can be further expanded, yielding
\begin{equation}\label{e17}
\frac{\partial \textbf{F}(\textbf{q},\textbf{p})}{\partial \textbf{q}}=\lim\limits_{\delta\rightarrow 0}\sum_{j=0}^mb_j^{(m)}\frac{\big[\textbf{p}^j\sum_{k=0}^n\binom{n}{k}\textbf{q}^{n-k}\delta^{k}\textbf{p}^{m-j}-\textbf{p}^j\textbf{q}^n\textbf{p}^{m-j}\big]}{\delta}
\end{equation}
where the non-vanishing terms are for all values of $j$ and $k=0,1$. Further simplification leads to
\begin{equation}\label{e5}
\frac{\partial \textbf{F}(\textbf{q},\textbf{p})}{\partial \textbf{q}}=n\sum_{j=0}^mb_j^{(m)}\textbf{p}^j\textbf{q}^{n-1}\textbf{p}^{m-j},
\end{equation}
which can be written as
\begin{equation}\label{e6}
\frac{\partial \textbf{F}(\textbf{q},\textbf{p})}{\partial \textbf{q}}=n\textbf{A}_{m,n-1}.
\end{equation}
\end{proof}
As we can see from the previous theorem, the differential quotient of second type works on observables obtained via any quantization rule. That is, there is a derivative with a well-defined value for operators obtained using Weyl, symmetric, or Born-Jordan ordering. However, this theorem only shows the value of the derivative for positive powers of position and momentum. A question now arises if this result is also compatible for operators with negative indices. We tackle this problem in the following theorem.
\begin{theorem}\label{theorem2}
Let $f(q,p)=p^{-m}q^n$ with its arbitrarily-quantized quantum image $\textbf{F}(\textbf{q},\textbf{p})=\textbf{A}_{-m,n}$. The partial derivative with respect to the momentum is given by
\begin{equation}\label{f3}
\frac{\partial f(q,p)}{\partial p}=-mp^{-m-1}q^n\rightarrow\frac{\partial \textbf{F}(\textbf{q},\textbf{p})}{\partial \textbf{p}}=-m\textbf{A}_{-m-1,n}
\end{equation}
\end{theorem}
\begin{proof}
We start with the operator
\begin{equation}\label{f4}
\textbf{F}(\textbf{q},\textbf{p})=\textbf{A}_{-m,n}=\sum_{j=0}^na_j^{(n)}\textbf{q}^j\textbf{p}^{-m}\textbf{q}^{n-j}
\end{equation}
where $a_j^{(n)}$ depends on the ordering rule.
Using the definition of partial derivative and extending it to operators, we obtain
\begin{equation}\label{f5}
\begin{split}
\frac{\partial \textbf{F}(\textbf{q},\textbf{p})}{\partial \textbf{p}}=\lim\limits_{\delta\rightarrow 0}\sum_{j=0}^na_j^{(n)}\frac{\big[\textbf{q}^j(\textbf{p}+\delta\textbf{1})^{-m}\textbf{q}^{n-j}-\textbf{q}^j\textbf{p}^{-m}\textbf{q}^{n-j}\big]}{\delta}.
\end{split}
\end{equation}
where the terms in the numerator can be factored such that
\begin{equation}\label{f6}
\frac{\partial \textbf{F}(\textbf{q},\textbf{p})}{\partial \textbf{p}}=\lim\limits_{\delta\rightarrow 0}\sum_{j=0}^na_j^{(n)}\frac{\textbf{q}^j\big[(\textbf{p}+\delta\textbf{1})^{-m}-\textbf{p}^{-m}\big]\textbf{q}^{n-j}}{\delta}.
\end{equation}
Equation \eqref{f6} is easier to evaluate when it is expressed as follows
\begin{equation}\label{f7}
\frac{\partial \textbf{F}(\textbf{q},\textbf{p})}{\partial \textbf{p}}=\lim\limits_{\delta\rightarrow 0}\sum_{j=0}^n\frac{a_j^{(n)}}{\delta}\textbf{q}^j\bigg\{\bigg[\textbf{p}^{m}-(\textbf{p}+\delta\textbf{1})^{m}\bigg]\textbf{p}^{-m}(\textbf{p}+\delta\textbf{1})^{-m}\bigg\}\textbf{q}^{n-j}.
\end{equation}
Upon expanding the factor $(\textbf{p}+\delta\textbf{1})^m$, we have
\begin{equation}\label{f8}
\frac{\partial \textbf{F}(\textbf{q},\textbf{p})}{\partial \textbf{p}}=\lim\limits_{\delta\rightarrow 0}\sum_{j=0}^n\frac{a_j^{(n)}}{\delta}\textbf{q}^j\bigg\{\bigg[\textbf{p}^{m}-\sum_{k=0}^m\binom{m}{k}\textbf{p}^{m-k}\delta^k\bigg]\textbf{p}^{-m}(\textbf{p}+\delta\textbf{1})^{-m}\bigg\}\textbf{q}^{n-j},
\end{equation}
where the term $k=0$ in the bracket cancels out. This gives us
\begin{equation}\label{f9}
\frac{\partial \textbf{F}(\textbf{q},\textbf{p})}{\partial \textbf{p}}=\lim\limits_{\delta\rightarrow 0}\sum_{j=0}^n\frac{a_j^{(n)}}{\delta}\textbf{q}^j\bigg\{\bigg[-\sum_{k=1}^m\binom{m}{k}\textbf{p}^{m-k}\delta^k\bigg]\textbf{p}^{-m}(\textbf{p}+\delta\textbf{1})^{-m}\bigg\}\textbf{q}^{n-j}
\end{equation}
in which the non-vanishing terms are for all values of $j$ and for $k=1$ when the limit is evaluated. Therefore, we have
\begin{equation}\label{f10}
\frac{\partial \textbf{F}(\textbf{q},\textbf{p})}{\partial \textbf{p}}=-m\sum_{j=0}^na_j^{(n)}\textbf{q}^j\textbf{p}^{-m-1}\textbf{q}^{n-j}
\end{equation}
which can be written as
\begin{equation}\label{f11}
\frac{\partial \textbf{F}(\textbf{q},\textbf{p})}{\partial \textbf{p}}=-m\textbf{A}_{-m-1,n}.
\end{equation}
\end{proof}
\begin{theorem}\label{theorem3}
Let $f(q,p)=p^{m}q^{-n}$ with its arbitrarily-quantized quantum image $\textbf{F}(\textbf{q},\textbf{p})=\textbf{A}_{m,-n}$. The partial derivative with respect to the position is given by
\begin{equation}\label{f12}
\frac{\partial f(q,p)}{\partial q}=-np^{m}q^{-n-1}\rightarrow\frac{\partial \textbf{F}(\textbf{q},\textbf{p})}{\partial \textbf{q}}=-n\textbf{A}_{m,-n-1}.
\end{equation}
\end{theorem}
\begin{proof}
Starting with the operator given by
\begin{equation}\label{f13}
\textbf{F}(\textbf{q},\textbf{p})=\textbf{A}_{m,-n}=\sum_{j=0}^mb_j^{(m)}\textbf{p}^j\textbf{q}^{-n}\textbf{p}^{m-j}
\end{equation}
where $b_j^{(m)}$ depends on the ordering rule, along with the limit definition of the derivative, we arrive at
\begin{equation}\label{f14}
\begin{split}
\frac{\partial \textbf{F}(\textbf{q},\textbf{p})}{\partial \textbf{q}}=\lim\limits_{\delta\rightarrow 0}\sum_{j=0}^mb_j^{(m)}\frac{\big[\textbf{p}^j(\textbf{q}+\delta\textbf{1})^{-n}\textbf{p}^{m-j}-\textbf{p}^j\textbf{q}^{-n}\textbf{p}^{m-j}\big]}{\delta}
\end{split}
\end{equation}
which can be written as
\begin{equation}\label{f15}
\frac{\partial \textbf{F}(\textbf{q},\textbf{p})}{\partial \textbf{q}}=\lim\limits_{\delta\rightarrow 0}\sum_{j=0}^mb_j^{(m)}\frac{\textbf{p}^j\big[(\textbf{q}+\delta\textbf{1})^{-n}-\textbf{q}^{-n}\big]\textbf{p}^{m-j}}{\delta}.
\end{equation}
Further expansion of equation \eqref{f15} leads to
\begin{equation}\label{f16}
\frac{\partial \textbf{F}(\textbf{q},\textbf{p})}{\partial \textbf{q}}=\lim\limits_{\delta\rightarrow 0}\sum_{j=0}^m\frac{b_j^{(m)}}{\delta}\textbf{p}^j\bigg\{\bigg[\textbf{q}^n-(\textbf{q}+\delta\textbf{1})^n\bigg]\textbf{q}^{-n}(\textbf{q}+\delta\textbf{1})^{-n}\bigg\}\textbf{p}^{m-j},
\end{equation}
where the terms inside the bracket can be simplified. It follows that
\begin{equation}\label{f17}
\frac{\partial \textbf{F}(\textbf{q},\textbf{p})}{\partial \textbf{q}}=\lim\limits_{\delta\rightarrow 0}\sum_{j=0}^m\frac{b_j^{(m)}}{\delta}\textbf{p}^j\bigg\{\bigg[-\sum_{k=1}^n\binom{n}{k}\textbf{q}^{n-k}\delta^k\bigg]\textbf{q}^{-n}(\textbf{q}+\delta\textbf{1})^{-n}\bigg\}\textbf{p}^{m-j}
\end{equation}
where the non-vanishing terms are for all values of $j$ and for $k=1$. We arrive at
\begin{equation}\label{f18}
\frac{\partial \textbf{F}(\textbf{q},\textbf{p})}{\partial \textbf{q}}=-n\sum_{j=0}^mb_j^{(m)}\textbf{p}^j\textbf{q}^{-n-1}\textbf{p}^{m-j},
\end{equation}
which can be expressed as
\begin{equation}\label{f19}
\frac{\partial \textbf{F}(\textbf{q},\textbf{p})}{\partial \textbf{q}}=-n\textbf{A}_{m,-n-1}.
\end{equation}
\end{proof}
At this point we can see that no specific quantization is preferred by the differential quotient of second type. Let us recall the first definition of the derivative and its modification for the different quantizations in Section \ref{2.2} given by 
\begin{equation}\label{add21}
\begin{split}
    &\bigg(\frac{\partial \textbf{T}_{m,n}}{\partial \textbf{q}}\bigg)_W=n\textbf{T}_{m,n-1} \quad\quad\text{and}\quad\quad\bigg(\frac{\partial \textbf{T}_{m,n}}{\partial \textbf{p}}\bigg)_W=m\textbf{T}_{m-1,n},
    \\&\bigg(\frac{\partial \textbf{S}_{m,n}}{\partial \textbf{q}}\bigg)_S\;\;=n\textbf{S}_{m,n-1} \quad\quad\text{and}\quad\quad\bigg(\frac{\partial \textbf{S}_{m,n}}{\partial \textbf{p}}\bigg)_S\;\;=m\textbf{S}_{m-1,n},
    \\&\bigg(\frac{\partial \textbf{B}_{m,n}}{\partial \textbf{q}}\bigg)_1\;\;=n\textbf{B}_{m,n-1} \quad\quad\text{and}\quad\quad\bigg(\frac{\partial \textbf{B}_{m,n}}{\partial \textbf{p}}\bigg)_1\;\;=m\textbf{B}_{m-1,n}.
\end{split}
\end{equation}
Noting that the results of Theorem \ref{theorem1} apply for arbitrarily-ordered operators, we can observe the consistency of the differential quotient of first type and the differential quotient of second type. That is,
\begin{equation}\label{add22}
\begin{split}
    &\bigg(\frac{\partial \textbf{T}_{m,n}}{\partial \textbf{q}}\bigg)_W=\bigg(\frac{\partial \textbf{T}_{m,n}}{\partial \textbf{q}}\bigg)_2 \quad\quad\text{and}\quad\quad\bigg(\frac{\partial \textbf{T}_{m,n}}{\partial \textbf{p}}\bigg)_W=\bigg(\frac{\partial \textbf{T}_{m,n}}{\partial \textbf{p}}\bigg)_2
\end{split}
\end{equation}
for the Weyl quantization, then we have
\begin{equation}\label{add23}
\begin{split}
    &\bigg(\frac{\partial \textbf{S}_{m,n}}{\partial \textbf{q}}\bigg)_S=\bigg(\frac{\partial \textbf{S}_{m,n}}{\partial \textbf{q}}\bigg)_2 \quad\quad\text{and}\quad\quad\bigg(\frac{\partial \textbf{S}_{m,n}}{\partial \textbf{p}}\bigg)_S=\bigg(\frac{\partial \textbf{S}_{m,n}}{\partial \textbf{p}}\bigg)_2
\end{split}
\end{equation}
for the simplest symmetric quantization, and
\begin{equation}\label{add24}
\begin{split}
    &\bigg(\frac{\partial \textbf{B}_{m,n}}{\partial \textbf{q}}\bigg)_1=\bigg(\frac{\partial \textbf{B}_{m,n}}{\partial \textbf{q}}\bigg)_2 \quad\quad\text{and}\quad\quad\bigg(\frac{\partial \textbf{B}_{m,n}}{\partial \textbf{p}}\bigg)_1=\bigg(\frac{\partial \textbf{B}_{m,n}}{\partial \textbf{p}}\bigg)_2
\end{split}
\end{equation}
for the Born-Jordan quantization. What we have shown here is that the condition imposed in \cite{bj1926b} regarding the equivalence of the two definitions of the derivative can be achieved, not just for Born-Jordan quantization but also for the Weyl and simplest symmetric. Recalling the modifications, the differential quotient of first type is given by
\begin{equation}\label{aj47}
\bigg(\frac{\partial \textbf{y}}{\partial \textbf{y}_k}\bigg)_{W}=\frac{n}{2^{n-1}}\sum_{r=1}^s\delta_{l_r,k}\binom{n-1}{l-1}\prod_{m=r+1}^s\textbf{y}_{l_m}\prod_{m=1}^{r-1}\textbf{y}_{l_m}
\end{equation}
for the Weyl basis where
\begin{equation}\label{aj48}
n=\sum_{r=1}^s\delta_{l_r,k}
\end{equation}
and
\begin{equation}\label{aj49}
l=\sum_{t=1}^r\delta_{l_t,k}.
\end{equation}
Then we have
\begin{equation}\label{aj61}
\bigg(\frac{\partial \textbf{y}}{\partial \textbf{y}_k}\bigg)_{S}=\frac{n}{2}\sum_{r=1}^s\delta_{l_r,k}(\delta_{l,1}+\delta_{l,m})\prod_{m=r+1}^s\textbf{y}_{l_m}\prod_{m=1}^{r-1}\textbf{y}_{l_m}
\end{equation}
for the symmetric basis, and
\begin{equation}\label{aj2}
\bigg(\frac{\partial\textbf{y}}{\partial \textbf{y}_k}\bigg)_1=\sum_{r=1}^s\delta_{l_r,k}\prod_{m=r+1}^s\textbf{y}_{l_m}\prod_{m=1}^{r-1}\textbf{y}_{l_m}.
\end{equation}
for the Born-Jordan basis.
\\
\indent
Now that we have shown some theorems on the derivative of operators using the limit definition and its consistency with the differential quotient of first type, we now proceed to multiple differentiation.
\begin{theorem}\label{theorem5}
For  an arbitrarily-quantized operator $\textbf{A}_{m,n}$, multiple differentiation with respect to the momentum yields
\begin{equation}\label{f20}
\frac{\partial^s \textbf{A}_{m,n}}{\partial \textbf{p}^s}=\frac{m!}{(m-s)!}\textbf{A}_{m-s,n}.
\end{equation}
where $s\in\mathbb{Z}_+$.
\end{theorem}
\begin{proof}
For this instance we use mathematical induction. From Theorem \ref{theorem1}, we have the base case
\begin{equation}\label{f21}
\frac{\partial \textbf{A}_{m,n}}{\partial \textbf{p}}=m\textbf{A}_{m-1,n}.
\end{equation}
For the induction step, we let $u\in\mathbb{Z}_+$ and assume that \eqref{f20} is true for $s=u$. Then for $s=u+1$, we have
\begin{equation}\label{f22}
\frac{\partial^{u+1} \textbf{A}_{m,n}}{\partial \textbf{p}^{u+1}}=\frac{\partial}{\partial \textbf{p}}\bigg(\frac{m!}{(m-u)!}\textbf{A}_{m-u,n}\bigg).
\end{equation}
Next, we again use the limit definition of the derivative to arrive at
\begin{equation}\label{f23}
\frac{\partial^{u+1} \textbf{A}_{m,n}}{\partial \textbf{p}^{u+1}}=\frac{m!}{(m-u)!}\lim\limits_{\delta\rightarrow 0}\sum_{j=0}^{n}\frac{a_j^{(n)}}{\delta}\textbf{q}^j\bigg\{(\textbf{p}+\delta\textbf{1})^{m-u}-\textbf{p}^{m-u}\bigg\}\textbf{q}^{n-j}
\end{equation}
where the terms in the curly bracket can be expanded such that 
\begin{equation}\label{f24}
\frac{\partial^{u+1} \textbf{A}_{m,n}}{\partial \textbf{p}^{u+1}}=\frac{m!}{(m-u)!}\lim\limits_{\delta\rightarrow 0}\sum_{j=0}^{n}\frac{a_j^{(n)}}{\delta}\textbf{q}^j\bigg\{\sum_{k=1}^{m-u}\binom{m-u}{k}\textbf{p}^{m-u-k}\delta^k\bigg\}\textbf{q}^{n-j}.
\end{equation}
The non-vanishing terms in equation \eqref{f24} are for all values of $j$ and for $k=1$. This leads to
\begin{equation}\label{f25}
\frac{\partial^{u+1} \textbf{A}_{m,n}}{\partial \textbf{p}^{u+1}}=\frac{m!}{(m-u)!}(m-u)\sum_{j=0}^{n}a_j^{(n)}\textbf{q}^j\textbf{p}^{m-u-1}\textbf{q}^{n-j}
\end{equation}
which can be simplified as
\begin{equation}\label{f26}
\frac{\partial^{u+1} \textbf{A}_{m,n}}{\partial \textbf{p}^{u+1}}=\frac{m!}{(m-u-1)!}\textbf{A}_{m-u-1,n}.
\end{equation}
Thus, \eqref{f20} is true for $s=u+1$ and by the virtue of mathematical induction, \eqref{f20} is true for all $s\in\mathbb{Z}_+$.
\end{proof}
Next for multiple differentiation with respect to the position operator, we have the following Theorem.
\noindent
\begin{theorem}\label{theorem6}
For  an arbitrarily-quantized operator $\textbf{A}_{m,n}$, multiple differentiation with respect to the position is given by
\begin{equation}\label{f27}
\frac{\partial^t \textbf{A}_{m,n}}{\partial \textbf{q}^t}=\frac{n!}{(n-t)!}\textbf{A}_{m,n-t}.
\end{equation}
where $t\in\mathbb{Z}_+$.
\end{theorem}
\begin{proof}
We use mathematical induction. From Theorem \ref{theorem1}, we have the base case
\begin{equation}\label{f28}
\frac{\partial \textbf{A}_{m,n}}{\partial \textbf{q}}=n\textbf{A}_{m,n-1}.
\end{equation}
For the induction step, we now let $u\in\mathbb{Z}_+$ and assume that \eqref{f27} is true for $t=u$. Then for $t=u+1$, we have
\begin{equation}\label{f29}
\frac{\partial^{u+1} \textbf{A}_{m,n}}{\partial \textbf{q}^{u+1}}=\frac{\partial}{\partial \textbf{q}}\bigg(\frac{n!}{(n-u)!}\textbf{A}_{m,n-u}\bigg).
\end{equation}
Again, we use the limit definition of the derivative to obtain
\begin{equation}\label{f30}
\frac{\partial^{u+1} \textbf{A}_{m,n}}{\partial \textbf{q}^{u+1}}=\frac{n!}{(n-u)!}\lim\limits_{\delta\rightarrow 0}\sum_{j=0}^{m}\frac{b_j^{(m)}}{\delta}\textbf{p}^j\bigg\{(\textbf{q}+\delta\textbf{1})^{n-u}-\textbf{q}^{n-u}\bigg\}\textbf{p}^{m-j}
\end{equation}
which can be expanded as
\begin{equation}\label{f31}
\frac{\partial^{u+1} \textbf{A}_{m,n}}{\partial \textbf{q}^{u+1}}=\frac{n!}{(n-u)!}\lim\limits_{\delta\rightarrow 0}\sum_{j=0}^{m}\frac{b_j^{(m)}}{\delta}\textbf{p}^j\bigg\{\sum_{k=1}^{n-u}\binom{n-u}{k}\textbf{q}^{n-u-k}\delta^k\bigg\}\textbf{p}^{m-j}.
\end{equation}
The non-vanishing terms are for all values of $j$ and for $k=1$. This results in
\begin{equation}\label{f32}
\frac{\partial^{u+1} \textbf{A}_{m,n}}{\partial \textbf{q}^{u+1}}=\frac{n!}{(n-u)!}(n-u)\sum_{j=0}^{m}b_j^{(m)}\textbf{p}^j\textbf{q}^{n-u-1}\textbf{p}^{m-j}
\end{equation}
which simplifies into
\begin{equation}\label{f33}
\frac{\partial^{u+1} \textbf{A}_{m,n}}{\partial \textbf{q}^{u+1}}=\frac{n!}{(n-u-1)!}\textbf{A}_{m,n-u-1}.
\end{equation}
Thus, \eqref{f27} is true for $t=u+1$ and by the virtue of mathematical induction, \eqref{f27} is true for all $t\in\mathbb{Z}_+$.
\end{proof}
Notice that differentiation acts like a lowering operator which can be continuously done until the index becomes zero. The question now is whether or not these results on multiple differentiation also apply to operators with negative index. Let us look at the result in \eqref{f33}. If the index $n$ becomes $-n$, we now have terms of factorial with negative arguments. This is ill-defined and thus we resort to extending the factorial to gamma functions where the negative arguments are now easier to handle. However, another problem arises when evaluating specifically the gamma function of negative integers. Generally, the gamma function diverges when evaluated at negative integers and zero. For instance if we have
\begin{equation}\label{add17}
    s=\frac{\Gamma(-n)}{\Gamma(-n-m)} \quad\quad\quad n,m\in\mathbb{N},
\end{equation}
we now encounter a quantity $s$ with an indeterminate form $\frac{\infty}{\infty}$. We address this using the reflection formula given by \cite{arfken}
\begin{equation}\label{af34}
    \Gamma(-z)=-\frac{\pi\csc(\pi z)}{\Gamma(z+1)}.
\end{equation}
Using \eqref{add17} as an example, we have
\begin{equation}\label{add18}
    \frac{\Gamma(-n)}{\Gamma(-n-m)}:=\lim_{z\rightarrow n}\frac{\Gamma(-z)}{\Gamma(-z-m)}
\end{equation}
which, using the reflection formula, can be expressed as
\begin{equation}\label{add19}
    \frac{\Gamma(-n)}{\Gamma(-n-m)}=\lim_{z\rightarrow n}\frac{-\pi/[\sin(\pi z)\Gamma(z+1)]}{-\pi/[\sin(\pi (z+m))\Gamma(z+m+1)]}.
\end{equation}
Instead of dealing with gamma functions of negative integers, using the reflection formula resulted in evaluating limits of sine terms. We rewrite \eqref{add19} to obtain
\begin{equation}\label{add20}
    \frac{\Gamma(-n)}{\Gamma(-n-m)}=\lim_{z\rightarrow n}\frac{\sin(\pi(z+m))\Gamma(z+m+1)}{\sin(\pi z)\Gamma(z+1)}
\end{equation}
where we use L'Hopital's rule to address the indeterminate form due to the sine terms. We now arrive at
\begin{equation}\label{add20a}
    \frac{\Gamma(-n)}{\Gamma(-n-m)}=\lim_{z\rightarrow n}\frac{\cos(\pi(z+m))\Gamma(z+m+1)}{\cos(\pi z)\Gamma(z+1)}
\end{equation}
which further simplifies as
\begin{equation}\label{add20b}
    \frac{\Gamma(-n)}{\Gamma(-n-m)}=\frac{(-1)^{n+m}(n+m)!}{(-1)^nn!}.
\end{equation}
We now have
\begin{equation}\label{add20c}
    \frac{\Gamma(-n)}{\Gamma(-n-m)}=(-1)^m\frac{(n+m)!}{n!}
\end{equation}
which is a well-defined quantity. 
\\
\indent
This result is important in evaluating multiple differentiation of operators with negative index. We now investigate this case in the succeeding Theorem.
\begin{theorem}\label{theorem7}
For  an arbitrarily-quantized operator $\textbf{A}_{-m,n}$, multiple differentiation with respect to the momentum gives
\begin{equation}\label{f34}
\frac{\partial^s \textbf{A}_{-m,n}}{\partial \textbf{p}^s}=(-1)^s\frac{(m+s-1)!}{(m-1)!}\textbf{A}_{-m-s,n}.
\end{equation}
where $s\in\mathbb{Z}_+$.
\end{theorem}
\begin{proof}
We use mathematical induction. From Theorem \ref{theorem2}, we have the base case
\begin{equation}\label{f35}
\frac{\partial \textbf{A}_{-m,n}}{\partial \textbf{p}}=-m\textbf{A}_{-m-1,n}.
\end{equation}
For the induction step, we let $u\in\mathbb{Z}_+$ and assume that \eqref{f34} is true for $s=u$. Then for $s=u+1$, 
\begin{equation}\label{f36}
\frac{\partial^{u+1} \textbf{A}_{-m,n}}{\partial \textbf{p}^{u+1}}=\frac{\partial}{\partial \textbf{p}}\bigg((-1)^u\frac{(m+u-1)!}{(m-1)!}\textbf{A}_{-m-u,n}\bigg).
\end{equation}
Using the limit definition of the derivative, we obtain
\begin{equation}\label{f37}
\begin{split}
\frac{\partial^{u+1} \textbf{A}_{-m,n}}{\partial \textbf{p}^{u+1}}=(-1)^u\frac{(m+u-1)!}{(m-1)!}\lim\limits_{\delta\rightarrow 0}\sum_{j=0}^n\frac{a_j^{(n)}}{\delta}\textbf{q}^j\bigg[(\textbf{p}+\delta\textbf{1})^{-m-u}-\textbf{p}^{-m-u}\bigg]\textbf{q}^{n-j}
\end{split}
\end{equation}
which can be written as
\begin{equation}\label{f38}
\begin{split}
\frac{\partial^{u+1} \textbf{A}_{-m,n}}{\partial \textbf{p}^{u+1}}=(-1)^u\frac{(m+u-1)!}{(m-1)!}&\lim\limits_{\delta\rightarrow 0}\sum_{j=0}^n\frac{a_j^{(n)}}{\delta}\textbf{q}^j\bigg\{\bigg[\textbf{p}^{m+u}-(\textbf{p}+\delta\textbf{1})^{m+u}\bigg]\\&\times\textbf{p}^{-m-u}(\textbf{p}+\delta\textbf{1})^{-m-u}\bigg\}\textbf{q}^{n-j}
\end{split}
\end{equation}
where the terms in the  bracket can be modified such that
\begin{equation}\label{f39}
\begin{split}
\frac{\partial^{u+1} \textbf{A}_{-m,n}}{\partial \textbf{p}^{u+1}}=(-1)^u\frac{(m+u-1)!}{(m-1)!}&\lim\limits_{\delta\rightarrow 0}\sum_{j=0}^n\frac{a_j^{(n)}}{\delta}\textbf{q}^j\bigg\{\bigg[-\sum_{k=1}^{m+u}\binom{m+u}{k}\textbf{p}^{m+u-k}\delta^k\bigg]\\&\times\textbf{p}^{-m-u}(\textbf{p}+\delta\textbf{1})^{-m-u}\bigg\}\textbf{q}^{n-j}.
\end{split}
\end{equation}
The non-vanishing terms in equation \eqref{f39} are for all values of $j$ and for $k=1$, which then gives us
\begin{equation}\label{f40}
\begin{split}
\frac{\partial^{u+1} \textbf{A}_{-m,n}}{\partial \textbf{p}^{u+1}}=-(-1)^u\frac{(m+u-1)!}{(m-1)!}(m+u)\sum_{j=0}^{n}a_j^{(n)}\textbf{q}^j\textbf{p}^{-m-u-1}\textbf{q}^{n-j}.
\end{split}
\end{equation}
This simplifies into
\begin{equation}\label{f41}
\begin{split}
\frac{\partial^{u+1} \textbf{A}_{-m,n}}{\partial \textbf{p}^{u+1}}=(-1)^{u+1}\frac{(m+u)!}{(m-1)!}\textbf{A}_{-m-u-1,n}.
\end{split}
\end{equation}
From \eqref{f41}, we can see that \eqref{f34} is true for $s=u+1$ and by the virtue of mathematical induction, \eqref{f34} is true for all $s\in\mathbb{Z}_+$.
\end{proof}
In the classical case, we can see that
\begin{equation}\label{extra1}
    \frac{\partial^s p^{-m}q^n}{\partial p^s}=(-m)(-m-1)\cdot\cdot\cdot(-m-s+1)p^{-m-s}q^n
\end{equation}
which can also be written as 
\begin{equation}\label{extra2}
    \frac{\partial^s p^{-m}q^n}{\partial p^s}=(-1)^{s}\frac{(m+s-1)!}{(m-1)!}p^{-m-s}q^n.
\end{equation}
When we try to obtain the quantized version of \eqref{extra2}, the result goes back to Theorem \ref{theorem7} which shows its consistency even when starting from the classical case. Extending the equation in Theorem \ref{theorem2} to multiple differentiation, we have 
\begin{equation}\label{extra3}
\frac{\partial^s p^{-m}q^n}{\partial p^s}=(-1)^{s}\frac{(m+s-1)!}{(m-1)!}p^{-m-s}q^n\rightarrow\frac{\partial^s \textbf{A}_{-m,n}}{\partial \textbf{p}^s}=(-1)^s\frac{(m+s-1)!}{(m-1)!}\textbf{A}_{-m-s,n}.
\end{equation}
\begin{theorem}\label{theorem8}
For  an arbitrarily-quantized operator $\textbf{A}_{m,-n}$, multiple differentiation with respect to the position operator yields
\begin{equation}\label{f42}
\frac{\partial^t \textbf{A}_{m,-n}}{\partial \textbf{q}^t}=(-1)^t\frac{(n+t-1)!}{(n-1)!}\textbf{A}_{m,-n-t}.
\end{equation}
where $t\in\mathbb{Z}_+$.
\end{theorem}
\begin{proof}
We use mathematical induction. From Theorem \ref{theorem3}, the base case is given by
\begin{equation}\label{f43}
\frac{\partial \textbf{A}_{m,-n}}{\partial \textbf{q}}=-n\textbf{A}_{m,-n-1}.
\end{equation}
Next, we let $u\in\mathbb{Z}_+$ and assume that \eqref{f42} is true for $t=u$. Then for $t=u+1$, 
\begin{equation}\label{f44}
\frac{\partial^{u+1} \textbf{A}_{m,-n}}{\partial \textbf{q}^{u+1}}=\frac{\partial}{\partial \textbf{q}}\bigg((-1)^u\frac{(n+u-1)!}{(n-1)!}\textbf{A}_{m,-n-u}\bigg).
\end{equation}
From the limit definition of the derivative, we have
\begin{equation}\label{f45}
\begin{split}
\frac{\partial^{u+1} \textbf{A}_{m,-n}}{\partial \textbf{q}^{u+1}}=(-1)^u\frac{(n+u-1)!}{(n-1)!}\lim\limits_{\delta\rightarrow 0}\sum_{j=0}^m\frac{b_j^{(m)}}{\delta}\textbf{p}^j\bigg[(\textbf{q}+\delta\textbf{1})^{-n-u}-\textbf{q}^{-n-u}\bigg]\textbf{p}^{m-j}
\end{split}
\end{equation}
which can be expressed as
\begin{equation}\label{f46}
\begin{split}
\frac{\partial^{u+1} \textbf{A}_{m,-n}}{\partial \textbf{q}^{u+1}}=(-1)^u\frac{(n+u-1)!}{(n-1)!}&\lim\limits_{\delta\rightarrow 0}\sum_{j=0}^m\frac{b_j^{(m)}}{\delta}\textbf{p}^j\bigg\{\bigg[\textbf{q}^{n+u}-(\textbf{q}+\delta\textbf{1})^{n+u}\bigg]\\&\times\textbf{q}^{-n-u}(\textbf{q}+\delta\textbf{1})^{-n-u}\bigg\}\textbf{p}^{m-j}
\end{split}
\end{equation}
where the terms in the  bracket can be rewritten such that
\begin{equation}\label{f47}
\begin{split}
\frac{\partial^{u+1} \textbf{A}_{m,-n}}{\partial \textbf{q}^{u+1}}=(-1)^u\frac{(n+u-1)!}{(n-1)!}&\lim\limits_{\delta\rightarrow 0}\sum_{j=0}^m\frac{b_j^{(m)}}{\delta}\textbf{p}^j\bigg\{\bigg[-\sum_{k=1}^{n+u}\binom{n+u}{k}\textbf{q}^{n+u-k}\delta^k\bigg]\\&\times\textbf{q}^{-n-u}(\textbf{q}+\delta\textbf{1})^{-n-u}\bigg\}\textbf{p}^{m-j}.
\end{split}
\end{equation}
From equation \eqref{f47}, the non-vanishing terms are for all values of $j$ and for $k=1$, which then yields
\begin{equation}\label{f48}
\begin{split}
\frac{\partial^{u+1} \textbf{A}_{m,-n}}{\partial \textbf{q}^{u+1}}=-(-1)^u\frac{(n+u-1)!}{(n-1)!}(n+u)\sum_{j=0}^{m}b_j^{(m)}\textbf{p}^j\textbf{q}^{-n-u-1}\textbf{p}^{m-j}.
\end{split}
\end{equation}
Equation \eqref{f48} can be further simplified as
\begin{equation}\label{f49}
\begin{split}
\frac{\partial^{u+1} \textbf{A}_{m,-n}}{\partial \textbf{q}^{u+1}}=(-1)^{u+1}\frac{(n+u)!}{(n-1)!}\textbf{A}_{m,-n-u-1}.
\end{split}
\end{equation}
From \eqref{f49}, we can see that \eqref{f42} is true for $t=u+1$. Then, by the virtue of mathematical induction, \eqref{f42} is true for all $t\in\mathbb{Z}_+$.
\end{proof}
Going to the classical case, we have
\begin{equation}\label{extra4}
    \frac{\partial^t p^{m}q^{-n}}{\partial q^t}=(-n)(-n-1)\cdot\cdot\cdot(-n-t+1)p^{m}q^{-n-t}
\end{equation}
which is equal to 
\begin{equation}\label{extra5}
    \frac{\partial^t p^{m}q^{-n}}{\partial q^t}=(-1)^{t}\frac{(n+t-1)!}{(n-1)!}p^{m}q^{-n-t}.
\end{equation}
When we try to obtain the quantized version of \eqref{extra5}, the result goes back to Theorem \ref{theorem8} which shows its consistency even when starting from the classical case. We extend the equation in Theorem \ref{theorem3} to multiple differentiation to obtain 
\begin{equation}\label{extra6}
\frac{\partial^t p^{m}q^{-n}}{\partial q^t}=(-1)^{t}\frac{(n+t-1)!}{(n-1)!}p^{m}q^{-n-t}\rightarrow\frac{\partial^t \textbf{A}_{m,-n}}{\partial \textbf{q}^t}=(-1)^t\frac{(n+t-1)!}{(n-1)!}\textbf{A}_{m,-n-t}.
\end{equation}

We now established some differentiation rules for operators with arbitrary ordering. From the Theorems, we can see that they almost resemble the differentiation rules we have in ordinary calculus. Even though we are working with operators, and adding the fact that they do not commute, we still see some similarities when dealing with quantum operators and going back to the classical regime. Since the non-commutativity of $\textbf{q}$ and $\textbf{p}$ was mentioned, we now investigate the differentiation rules with respect to multiple operators. Using the previous results, we present the subsequent theorem.
\begin{theorem}\label{theorem9}
The order of differentiation does not affect the resulting partial derivative. That is, for an arbitrarily-ordered operator $\textbf{A}_{m,n}$,
\begin{equation}\label{f50}
\frac{\partial^{s+t}\textbf{A}_{m,n}}{\partial\textbf{q}^t\partial\textbf{p}^s}=\frac{\partial^{s+t}\textbf{A}_{m,n}}{\partial\textbf{p}^s\partial\textbf{q}^t}
\end{equation}
for $s,t\in\mathbb{Z}_+$.
\end{theorem}
\begin{proof}
Starting from Theorem \ref{theorem5}, we have
\begin{equation}\label{f51}
\frac{\partial ^s \textbf{A}_{m,n}}{\partial \textbf{p}^s}=\frac{m!}{(m-s)!}\textbf{A}_{m-s,n}.
\end{equation}
We now take the derivative of both sides with respect to $\textbf{q}$, $t$ times. This results in
\begin{equation}\label{f52}
\frac{\partial^{s+t}\textbf{A}_{m,n}}{\partial\textbf{q}^t\partial\textbf{p}^s}=\frac{m!}{(m-s)!}\frac{\partial^t\textbf{A}_{m-s,n}}{\partial \textbf{q}^t}.
\end{equation}
Using the result of Theorem \ref{theorem6}, we can simplify the right hand side of equation \eqref{f52} as
\begin{equation}\label{f53}
\frac{\partial^{s+t}\textbf{A}_{m,n}}{\partial\textbf{q}^t\partial\textbf{p}^s}=\frac{m!}{(m-s)!}\frac{n!}{(n-t)!}\textbf{A}_{m-s,n-t}.
\end{equation}
We again use the result of Theorem \ref{theorem6} given by
\begin{equation}\label{f54}
\frac{\partial^t \textbf{A}_{m,n}}{\partial \textbf{q}^t}=\frac{n!}{(n-t)!}\textbf{A}_{m,n-t}.
\end{equation}
Differentiating both sides with respect to $\textbf{p}$, $s$ times, we have
\begin{equation}\label{f55}
\frac{\partial^{s+t} \textbf{A}_{m,n}}{\partial \textbf{p}^s\partial \textbf{q}^t}=\frac{n!}{(n-t)!}\frac{\partial ^s\textbf{A}_{m,n-t}}{\partial \textbf{p}^s}.
\end{equation}
Using Theorem \ref{theorem5}, the right hand side simplifies which is given by
\begin{equation}\label{f56}
\frac{\partial^{s+t} \textbf{A}_{m,n}}{\partial \textbf{p}^s\partial \textbf{q}^t}=\frac{n!}{(n-t)!}\frac{m!}{(m-s)!}\textbf{A}_{m-s,n-t}.
\end{equation}
Notice that both the right hand sides of equations \eqref{f53} and \eqref{f56} are equal. Therefore, both their left hand sides should also equate. Finally, we have
\begin{equation}\label{f57}
\frac{\partial^{s+t}\textbf{A}_{m,n}}{\partial\textbf{q}^t\partial\textbf{p}^s}=\frac{\partial^{s+t}\textbf{A}_{m,n}}{\partial\textbf{p}^s\partial\textbf{q}^t}.
\end{equation}
\end{proof}
This observation is quite remarkable since $\textbf{q}$ and $\textbf{p}$ are non-commuting in the first place. It is somehow expected that changing the order of differentiation with respect to these operators leads to a different result. However, as we have shown, the order of $\textbf{q}$ and $\textbf{p}$ does not matter when dealing with differentiating an operator with respect to another operator.
\section{Quantum Equations of Motion for Weyl, Simplest Symmetric, and Born-Jordan Quantized Operators}\label{2.4}
Now that we have shed some light about the derivative of an operator, let us go back to the quantum equations of motion since this is where the issue came from. We recall that
\begin{equation}\label{f58}
\begin{split}
&[\textbf{H},\textbf{q}]=-i\hbar\frac{\partial \textbf{H}}{\partial \textbf{p}}\quad\quad\quad\text{and}\quad\quad\quad[\textbf{H},\textbf{p}]=i\hbar\frac{\partial \textbf{H}}{\partial \textbf{q}}.
\end{split}
\end{equation}
Suppose the classical Hamiltonian has the form $H(q,p)=p^mq^n$ and its arbitrarily-quantized operator counterpart is $\textbf{H}=\textbf{A}_{m,n}$. Starting with the first commutator, we have
\begin{equation}\label{p16}
    [\textbf{A}_{m,n},\textbf{q}]=\sum_{j=0}^na_j^{(n)}\textbf{q}^j\textbf{p}^m\textbf{q}^{n-j+1}-\sum_{j=0}^na_j^{(n)}\textbf{q}^{j+1}\textbf{p}^m\textbf{q}^{n-j}
\end{equation}
where we use the relation
\begin{equation}\label{p17}
    \textbf{q}\textbf{p}^m=\textbf{p}^m\textbf{q}+i\hbar m\textbf{p}^{m-1}
\end{equation}
to rewrite the second term in the RHS of \eqref{p16}. This yields
\begin{equation}\label{p18}
\begin{split}
    [\textbf{A}_{m,n},\textbf{q}]=\sum_{j=0}^na_j^{(n)}\textbf{q}^j\textbf{p}^m\textbf{q}^{n-j+1}&-\sum_{j=0}^na_j^{(n)}\textbf{q}^{j}\textbf{p}^m\textbf{q}^{n-j+1}\\&-i\hbar m\sum_{j=0}^na_j^{(n)}\textbf{q}^{j}\textbf{p}^{m-1}\textbf{q}^{n-j}
\end{split}
\end{equation}
which simplifies as
\begin{equation}\label{p19}
\begin{split}
    [\textbf{A}_{m,n},\textbf{q}]=-i\hbar m\sum_{j=0}^na_j^{(n)}\textbf{q}^{j}\textbf{p}^{m-1}\textbf{q}^{n-j}
\end{split}
\end{equation}
or equivalently,
\begin{equation}\label{p20}
\begin{split}
    [\textbf{A}_{m,n},\textbf{q}]=-i\hbar m\textbf{A}_{m-1,n}.
\end{split}
\end{equation}
For the next quantum equation of motion, we have
\begin{equation}\label{p21}
    [\textbf{A}_{m,n},\textbf{p}]=\sum_{j=0}^mb_j^{(m)}\textbf{p}^j\textbf{q}^n\textbf{p}^{m-j+1}-\sum_{j=0}^mb_j^{(m)}\textbf{p}^{j+1}\textbf{q}^n\textbf{p}^{m-j}.
\end{equation}
We use the relation given by
\begin{equation}\label{p22}
    \textbf{p}\textbf{q}^n=\textbf{q}^n\textbf{p}-i\hbar n\textbf{q}^{n-1}
\end{equation}
to expand the second term in the RHS of \eqref{p21}. We now have
\begin{equation}\label{p23}
\begin{split}
    [\textbf{A}_{m,n},\textbf{p}]=&\sum_{j=0}^mb_j^{(m)}\textbf{p}^j\textbf{q}^n\textbf{p}^{m-j+1}-\sum_{j=0}^mb_j^{(m)}\textbf{p}^j\textbf{q}^n\textbf{p}^{m-j+1}\\&+i\hbar n\sum_{j=0}^nb_j^{(m)}\textbf{p}^{j}\textbf{q}^{n-1}\textbf{p}^{m-j}
\end{split}
\end{equation}
which can be simplified as
\begin{equation}\label{p24}
\begin{split}
    [\textbf{A}_{m,n},\textbf{p}]=i\hbar n\sum_{j=0}^mb_j^{(m)}\textbf{p}^{j}\textbf{q}^{n-1}\textbf{p}^{m-j}.
\end{split}
\end{equation}
We now write
\begin{equation}\label{p25}
\begin{split}
    [\textbf{A}_{m,n},\textbf{p}]=i\hbar n\textbf{A}_{m,n-1}.
\end{split}
\end{equation}
From Theorem \ref{theorem1}, we recall that
\begin{equation}\label{p26}
    \frac{\partial \textbf{A}_{m,n}}{\partial \textbf{p}}=m\textbf{A}_{m-1,n}
\end{equation}
and
\begin{equation}\label{p27}
    \frac{\partial \textbf{A}_{m,n}}{\partial \textbf{q}}=n\textbf{A}_{m,n-1}
\end{equation}
When we compare equations \eqref{p20} and \eqref{p26}, we can see that 
\begin{equation}\label{p28}
    [\textbf{A}_{m,n},\textbf{q}]=-i\hbar\frac{\partial \textbf{A}_{m,n}}{\partial \textbf{p}}.
\end{equation}
Now if equations \eqref{p25} and \eqref{p27} are compared, we have
\begin{equation}\label{p29}
    [\textbf{A}_{m,n},\textbf{p}]=i\hbar\frac{\partial \textbf{A}_{m,n}}{\partial \textbf{q}}.
\end{equation}
\\
\indent
From \eqref{p28} and \eqref{p29}, we can now conclude that arbitrarily quantizing the classical Hamiltonian, \textit{i.e.} $H\longrightarrow \textbf{H}=\textbf{A}_{m,n}$, satisfies the quantum equations of motion. Further expounding this, we now claim that the choice of ordering rule does not affect the validity of the resulting quantum equations of motion. Written individually, we have
\begin{equation}\label{p30}
\begin{split}
    &[\textbf{T}_{m,n},\textbf{q}]=-i\hbar\frac{\partial \textbf{T}_{m,n}}{\partial \textbf{p}}\quad\quad\text{and}\quad\quad    [\textbf{T}_{m,n},\textbf{p}]=i\hbar\frac{\partial \textbf{T}_{m,n}}{\partial \textbf{q}},
    \\&[\textbf{S}_{m,n},\textbf{q}]=-i\hbar\frac{\partial \textbf{S}_{m,n}}{\partial \textbf{p}}\;\quad\quad\text{and}\quad\quad    [\textbf{S}_{m,n},\textbf{p}]=i\hbar\frac{\partial \textbf{S}_{m,n}}{\partial \textbf{q}},
    \\&[\textbf{B}_{m,n},\textbf{q}]=-i\hbar\frac{\partial \textbf{B}_{m,n}}{\partial \textbf{p}}\quad\quad\text{and}\quad\quad    [\textbf{B}_{m,n},\textbf{p}]=i\hbar\frac{\partial \textbf{B}_{m,n}}{\partial \textbf{q}}.
\end{split}
\end{equation}
When using quantization, all the possible quantum images of the Hamiltonian will still exhibit the same quantum dynamical behavior. As we can see in \eqref{p30}, choosing either Weyl, simplest symmetric, or Born-Jordan ordering all lead to the correct quantum image, if our sole condition is to satisfy the quantum equations of motion.

To further expound the value of this result, let us apply it to the harmonic oscillator system. We recall that in the classical case, the Hamiltonian is
\begin{equation}\label{ee1}
    H=\frac{1}{2\mu}p^2+\frac{1}{2}\mu\omega^2q^2
\end{equation}
for a mass $\mu$ and frequency $\omega$. Using Weyl quantization, the quantum image of the Hamiltonian is
\begin{equation}\label{ee02}
    \textbf{H}_W=\frac{1}{2\mu}\textbf{T}_{2,0}+\frac{1}{2}\mu\omega^2\textbf{T}_{0,2}.
\end{equation}
For the position and momentum operators, we have the quantum images
\begin{equation}\label{ee2}
    q\mapsto \textbf{T}_{0,1} \quad\quad \text{and}\quad\quad p\mapsto \textbf{T}_{1,0}.
\end{equation}

We use the commutation relation for Weyl ordered forms \cite{domingo, bender3}
\begin{equation}\label{ee3}
\begin{split}
[\textbf{T}_{m,n}, \textbf{T}_{r,s}]=2\sum_{j=0}^{+\infty}&\frac{(i\hbar/2)^{2j+1}}{(2j+1)!}\sum_{k=0}^{2j+1}(-1)^k\binom{2j+1}{k}\frac{m!}{(m-k)!}\frac{n!}{(n-2j-1+k)!}  \\
&\times \frac{r!}{(r-2j-1+k)!}\frac{s!}{(s-k)!}\textbf{T}_{m+r-2j-1,n+s-2j-1},
\end{split}
\end{equation}
to simplify the commutators. For the commutator of the Hamiltonian and the position operator, we obtain
\begin{equation}\label{ee4}
    \frac{1}{2\mu}[\textbf{T}_{2,0}, \textbf{T}_{0,1}]+\frac{1}{2}\mu\omega^2[\textbf{T}_{0,2}, \textbf{T}_{0,1}]=\frac{-i\hbar}{\mu}\textbf{T}_{1,0}.
\end{equation}
The result \eqref{ee4} can be simplified into
\begin{equation}\label{ee5}
    [\textbf{H}_W,\textbf{q}]=\frac{-i\hbar}{\mu}\textbf{T}_{1,0}.
\end{equation}
Next, we calculate the partial derivative of the Weyl-quantized Hamiltonian with respect to the momentum operator. We have
\begin{equation}\label{ee6}
    \frac{\partial \textbf{H}_W}{\partial \textbf{p}}=\frac{1}{\mu}\textbf{T}_{1,0}.
\end{equation}
Comparing \eqref{ee5} and \eqref{ee6}, we can see by inspection that
\begin{equation}\label{ee7}
    [\textbf{H}_W,\textbf{q}]=-i\hbar\frac{\partial \textbf{H}_W}{\partial \textbf{p}}.
\end{equation}
Using \eqref{ee3}, we have the commutator
\begin{equation}\label{ee8}
    \frac{1}{2\mu}[\textbf{T}_{2,0}, \textbf{T}_{1,0}]+\frac{1}{2}\mu\omega^2[\textbf{T}_{0,2}, \textbf{T}_{1,0}]=\\mu\omega^2 i\hbar\textbf{T}_{0,1}
\end{equation}
which can be expressed as
\begin{equation}\label{ee9}
    [\textbf{H}_W,\textbf{p}]=\mu\omega^2i\hbar\textbf{T}_{0,1}.
\end{equation}
The partial derivative of the Hamiltonian with respect to the position operator is
\begin{equation}\label{ee10}
    \frac{\partial \textbf{H}_W}{\partial \textbf{q}}=\mu\omega^2\textbf{T}_{0,1}.
\end{equation}
From the results in \eqref{ee9} and \eqref{ee10}, it can be shown that
\begin{equation}\label{ee11}
    [\textbf{H}_W,\textbf{p}]=i\hbar\frac{\partial \textbf{H}_W}{\partial \textbf{q}}.
\end{equation}
We now have shown in \eqref{ee7} and \eqref{ee11} that Weyl quantization results to a Hamiltonian that is consistent with the quantum equations of motion.

When quantizing using the simplest symmetrization rule, we obtain the Hamiltonian
\begin{equation}\label{ee12}
    \textbf{H}_S=\frac{1}{2\mu}\textbf{S}_{2,0}+\frac{1}{2}\mu\omega^2\textbf{S}_{0,2}
\end{equation}
and we use the commutation relation for the basis operators given by \cite{domingo}
\begin{equation}\label{ee13}
\begin{split}
[\textbf{S}_{m,n}, \textbf{S}_{r,s}]&=2  \sum_{j=0}^{+\infty}\frac{(i\hbar/2)^{2j+1}}{(2j+1)!}  \sum_{k=0}^{2j+1}\binom{2j+1}{k}(-1)^k\sum_{l=0}^{+\infty}\frac{(i\hbar/2)^{2l}}{(2l)!}\sum_{t=0}^{+\infty}\frac{(i\hbar/2)^{2t}}{(2t)!}\\&\times \frac{m!}{(m-k-2l)!}\times\frac{n!}{(n-2j-1+k-2l)!}  \frac{r!}{(r-2j-1+k-2t)!} \\
&\times \frac{s!}{(s-k-2t)!}\sum_{u=0}^{+\infty}\frac{(-i\hbar/2)^uE_u}{u!}\frac{(m+r-2j-1-2l-2t)!}{(m+r-2j-1-2l-2t-u)!} \\
&\times
\frac{(n+s-2j-1-2l-2t)!}{(n+s-2j-1-2l-2t-u)!}\textbf{S}_{m+r-2j-1-2l-2t-u,n+s-2j-1-2l-2t-u},
\end{split}
\end{equation}
where $E_u$ are the Euler numbers of the first order. From here, we can have the commutators
\begin{equation}\label{ee14}
    \begin{split}
        &[\textbf{S}_{2,0}, \textbf{S}_{0,1}]=-2i\hbar\textbf{S}_{1,0} \\&
        [\textbf{S}_{0,2}, \textbf{S}_{0,1}]=0 \\&
        [\textbf{S}_{2,0}, \textbf{S}_{1,0}]=0 \\&
        [\textbf{S}_{0,2}, \textbf{S}_{1,0}]=2i\hbar\textbf{S}_{0,1}.
    \end{split}
\end{equation}
Therefore,
\begin{equation}\label{ee15}
    [\textbf{H}_S,\textbf{q}]=-\frac{i\hbar}{\mu}\textbf{S}_{1,0}
\end{equation}
and
\begin{equation}\label{ee16}
    [\textbf{H}_S,\textbf{p}]=\mu\omega^2i\hbar\textbf{S}_{0,1}.
\end{equation}
Taking the partial derivative of the Hamiltonian gives
\begin{equation}\label{ee17}
    \frac{\partial \textbf{H}_S}{\partial \textbf{p}}=\frac{1}{\mu}\textbf{S}_{1,0}.
\end{equation}
and
\begin{equation}\label{ee18}
    \frac{\partial \textbf{H}_S}{\partial \textbf{q}}=\mu\omega^2\textbf{S}_{0,1}.
\end{equation}
Comparing equations \eqref{ee15} and \eqref{ee17}, then \eqref{ee16} and \eqref{ee18}, we can see by inspection that
\begin{equation}\label{ee19}
    [\textbf{H}_S,\textbf{q}]=-i\hbar \frac{\partial \textbf{H}_S}{\partial \textbf{p}}
\end{equation}
and
\begin{equation}\label{ee20}
    [\textbf{H}_S,\textbf{p}]=i\hbar \frac{\partial \textbf{H}_S}{\partial \textbf{q}},
\end{equation}
which is consistent with the quantum analogue of the Hamilton's equations of motion.

Lastly if we quantize using the Born-Jordan ordering, we have the Hamiltonian
\begin{equation}\label{ee21}
    \textbf{H}_{BJ}=\frac{1}{2\mu}\textbf{B}_{2,0}+\frac{1}{2}\mu\omega^2\textbf{B}_{0,2}.
\end{equation}
The commutation relation for the Born-Jordan basis operators is \cite{domingo}
\begin{equation}\label{ee22}
\begin{split}
[\textbf{B}_{m,n}, \textbf{B}_{r,s}]&=2  \sum_{j=0}^{+\infty}\frac{(i\hbar/2)^{2j+1}}{(2j+1)!}  \sum_{k=0}^{2j+1}\binom{2j+1}{k}(-1)^k\sum_{l=0}^{+\infty}\frac{(i\hbar/2)^{2l}}{(2l+1)!}\sum_{t=0}^{+\infty}\frac{(i\hbar/2)^{2t}}{(2t+1)!}  \\ 
&\times\frac{m!}{(m-k-2l)!}\frac{n!}{(n-2j-1+k-2l)!}  \frac{r!}{(r-2j-1+k-2t)!} \\
&\times\frac{s!}{(s-k-2t)!} \sum_{u=0}^{+\infty}\frac{(-i\hbar)^uB_u(1/2)}{u!}\frac{(m+r-2j-1-2l-2t)!}{(m+r-2j-1-2l-2t-u)!} \\
&\times
\frac{(n+s-2j-1-2l-2t)!}{(n+s-2j-1-2l-2t-u)!}\textbf{B}_{m+r-2j-1-2l-2t-u,n+s-2j-1-2l-2t-u},
\end{split}
\end{equation}
where $B_u(1/2)$ are the Bernoulli polynomials evaluated at $1/2$. From here, we can have the commutators
\begin{equation}\label{ee23}
    \begin{split}
        &[\textbf{B}_{2,0}, \textbf{B}_{0,1}]=-2i\hbar\textbf{B}_{1,0} \\&
        [\textbf{B}_{0,2}, \textbf{B}_{0,1}]=0 \\&
        [\textbf{B}_{2,0}, \textbf{B}_{1,0}]=0 \\&
        [\textbf{B}_{0,2}, \textbf{B}_{1,0}]=2i\hbar\textbf{B}_{0,1}
    \end{split}
\end{equation}
which shows that
\begin{equation}\label{ee24}
    [\textbf{H}_{BJ},\textbf{q}]=-\frac{i\hbar}{\mu}\textbf{B}_{1,0}
\end{equation}
and
\begin{equation}\label{ee25}
    [\textbf{H}_{BJ},\textbf{p}]=\mu\omega^2i\hbar\textbf{B}_{0,1}.
\end{equation}
Taking the partial derivative of the Hamiltonian gives
\begin{equation}\label{ee26}
    \frac{\partial \textbf{H}_{BJ}}{\partial \textbf{p}}=\frac{1}{\mu}\textbf{B}_{1,0}.
\end{equation}
and
\begin{equation}\label{ee27}
    \frac{\partial \textbf{H}_{BJ}}{\partial \textbf{q}}=\mu\omega^2\textbf{B}_{0,1}.
\end{equation}
When we compare equations \eqref{ee24} and \eqref{ee26}, then \eqref{ee25} and \eqref{ee27}, it can be shown that
\begin{equation}\label{ee28}
    [\textbf{H}_{BJ},\textbf{q}]=-i\hbar \frac{\partial \textbf{H}_{BJ}}{\partial \textbf{p}}
\end{equation}
and
\begin{equation}\label{ee29}
    [\textbf{H}_{BJ},\textbf{p}]=i\hbar \frac{\partial \textbf{H}_{BJ}}{\partial \textbf{q}}.
\end{equation}
The result yields a set of quantum equations of motion where in this case, the Hamiltonian is quantized using the Born-Jordan ordering.

As shown in the harmonic oscillator example, we can arrive at the correct equations of motion regardless of the quantization rule being used. This supports our claim in \eqref{p30} that all the common quantization rules are valid choices when the goal is to arrive at the quantum analogue of the Hamilton's equations of motion.

\section{Conclusion}
We addressed the problem of finding the appropriate quantum image of a classical observable by constructing the quantum analogue of Hamilton's equations of motion.

From the fundamental works of Born and Jordan \cite{bj1925a, bj1926b}, which was used in the works of de Gosson \cite{j2degossonbook,j9degosson2014}, it was shown that the equivalence of the Schr{\"o}dinger and Heisenberg pictures of quantum mechanics takes place when mapping quantum operators from classical observables using the Born-Jordan quantization. This stems from the idea that the quantization of the classical Hamiltonian should result to an operator that obeys the similar dynamics governed by the Hamilton's equations of motion in classical mechanics.

To construct the quantum equations of motion, it is necessary to give meaning to a derivative of an operator with respect to another operator. The first definition called the differential quotient of first type defines the derivative as a linear combination of cyclic permutations of the factors in the operator being differentiated. This was the one originally imposed by Born and Jordan \cite{bj1925a} to construct the quantum equations of motion. This works for Hamiltonians being quantized using the Born-Jordan quantization but the consistency of this derivative does not hold for Weyl and symmetrically quantized Hamiltonians. In detail, we have shown in Section \ref{2.2} that this definition of the derivative can be modified such that the results are consistent for each basis. The symmetry of the amplitudes of each contributing terms in the derivatives play a role in the construction of the modifications. For the derivative that is consistent with the Born-Jordan quantization, the terms in the derivative all have coefficients equal to one. This results to a derivative that can be expressed in terms of the Born-Jordan basis operators $\textbf{B}_{m,n}$. We modified this definition for it to be expressed in terms of other basis operators. In the Weyl basis, the amplitudes of the resulting terms in the derivative were changed into the binomial coefficients. This way of defining the differential quotient of first type results to derivatives that can be expressed in terms of Weyl-ordered forms $\textbf{T}_{m,n}$. Another modification done was changing the coefficients into Kronecker deltas that are symmetrically distributed. This resulted to another definition of the derivative that result to operators that are in terms of $\textbf{S}_{m,n}$. A noteworthy remark about the three versions of the differential quotient of first type is that they are only valid for their respective basis. The original definition only works in the Born-Jordan basis, while the modifications done are only consistent in either Weyl or symmetric basis.

The differential quotient of second type defines the derivative as a limit, similar to the one being used in ordinary calculus. In Section \ref{2.3}, we started with the definition of the derivatives of operators using limits followed by proof of linearity of the differential quotient of second type. The theorems presented in this section deal with its properties which includes differentiation with respect to the position and momentum operators, differentiation of operators with negative index, and multiple differentiation. All these theorems apply to operators of arbitrary ordering, and is not dependent on a specific quantization rule. Unlike the differential quotient of first type, this definition of derivative is more general and it shows consistency with ordinary calculus. This supports the quantum-classical correspondence as we map the operators back into the classical phase space. Since the differential quotient of second type applies to operators of arbitrary ordering, the modification of the differential quotient of first type is necessary for the two definitions to be consistent with each other.

After establishing the validity of the two definitions of the derivative, the quantum version of the Hamilton's equations of motion was constructed in Section \ref{2.4}. We have shown that choosing either Weyl, simplest symmetric, or Born-Jordan quantization all lead to the correct quantum equations of motion. The dynamics of the operators follow the same behavior as the classical counterpart. As an example, the quantum harmonic oscillator Hamiltonian was used to show that choosing any of the three common quantization rules all lead to the correct equations of motion.

We note that all the differentiation rules that were established in this paper are only applicable for operators that follow the general form 
\begin{equation}\label{conc1}
\textbf{A}_{m,n}:=\sum_{j=0}^{n}a_j^{(n)}\textbf{q}^j\textbf{p}^m\textbf{q}^{n-j}
\end{equation}
for non-negative $n$, and
\begin{equation}\label{conc2}
\textbf{A}_{m,n}:=\sum_{j=0}^mb_j^{(m)}\textbf{p}^j\textbf{q}^n\textbf{p}^{m-j}.
\end{equation}
for non-negative $m$. One area that was not discussed in this paper is the quantization of operators with negative powers for both the position and momentum. A proper definition for the quantum image of $p^{-m}q^{-n}$ should first be established for each ordering rule before constructing the corresponding quantum equations of motion.  Another problem arises when dealing with operators of another form, for instance  all the momentum terms are placed to the left. The differential quotient of first type fails to yield a consistent derivative (see appendix for an example). We leave these problems open in the meantime. As long as the operators follow the form in \eqref{conc1} and \eqref{conc2}, we conclude in this paper that all the common quantization rules are of equal footing, and no quantization rule is superior over the other when constructing the quantum equations of motion.

%% chapter end
\section*{Appendix}\label{aaa}
In this appendix we investigate the resulting differential quotient of first type when the Hamiltonian is expressed in other forms, specifically when the momentum operators are placed such that they precede the position operators.

To start, let us consider a Hamiltonian given by
\begin{equation}\label{a27}
\textbf{H}=\textbf{q}\textbf{p}^m.
\end{equation}
Using the differential quotient of first type, the derivative with respect to the momentum operator is
\begin{equation}\label{aa28}
\begin{split}
\bigg(\frac{\partial\textbf{H}}{\partial \textbf{p}}\bigg)_1&=\textbf{p}^{m-1}\textbf{q}+\textbf{p}^{m-2}\textbf{q}\textbf{p}+...+\textbf{q}\textbf{p}^{m-1}
\end{split}
\end{equation}
which can be written as 
\begin{equation}\label{a28}
\begin{split}
\bigg(\frac{\partial\textbf{H}}{\partial \textbf{p}}\bigg)_1&=\sum_{j=0}^{m-1}\textbf{p}^j\textbf{q}\textbf{p}^{m-1-j}
\\&=m\textbf{B}_{m-1,1}.
\end{split}
\end{equation}
Now we rearrange the Hamiltonian such that the terms are antinormally ordered (momentum operators to the left). Using the canonical commutation relation $[\textbf{q},\textbf{p}]=i\hbar$, we obtain
\begin{equation}\label{a29}
\textbf{H}=\textbf{H}_{an}=\textbf{p}^m\textbf{q}+mi\hbar\textbf{p}^{m-1}.
\end{equation}
The derivative with respect to the momentum operator is then given by
\begin{equation}\label{aa30}
\bigg(\frac{\partial\textbf{H}_{an}}{\partial \textbf{p}}\bigg)_1=\textbf{p}^{m-1}\textbf{q}+\textbf{p}^{m-2}\textbf{q}\textbf{p}+...+\textbf{q}\textbf{p}^{m-1}+i\hbar m(m-1)\textbf{p}^{m-2}
\end{equation}
which is equal to
\begin{equation}\label{a30}
\begin{split}
\bigg(\frac{\partial\textbf{H}_{an}}{\partial \textbf{p}}\bigg)_1&=\sum_{j=0}^{m-1}\textbf{p}^j\textbf{q}\textbf{p}^{m-1-j}+i\hbar m(m-1)\textbf{p}^{m-2}
\\&=m\textbf{B}_{m-1,1}+i\hbar m(m-1)\textbf{p}^{m-2}.
\end{split}
\end{equation}
Comparing equations \eqref{a28} and \eqref{a30}, an interesting observation that can be drawn is that
\begin{equation}\label{a33}
\bigg(\frac{\partial\textbf{H}}{\partial \textbf{p}}\bigg)_1\neq\bigg(\frac{\partial\textbf{H}_{an}}{\partial \textbf{p}}\bigg)_1.
\end{equation}
Essentially $\textbf{H}$ and $\textbf{H}_{an}$ are equal and therefore their derivatives should also be equal. However, it can be seen in \eqref{a33} that this is not true when using the first definition of the derivative.

This observation can also be seen in the modifications of the differential quotient of first type. In the Weyl basis modification, we have
\begin{equation}\label{aa34}
\begin{split}
\bigg(\frac{\partial\textbf{H}}{\partial \textbf{p}}\bigg)_W&=\frac{m}{2^{m-1}}\bigg[\binom{m-1}{m-1}\textbf{p}^{m-1}\textbf{q}+\binom{m-1}{m-2}\textbf{p}^{m-2}\textbf{q}\textbf{p}+...+\binom{m-1}{0}\textbf{q}\textbf{p}^{m-1}\bigg]
\end{split}
\end{equation}
which simplifies as
\begin{equation}\label{aa35}
\bigg(\frac{\partial\textbf{H}}{\partial \textbf{p}}\bigg)_W=m\textbf{T}_{m-1,1}.
\end{equation}
Differentiating $\textbf{H}_{an}$, the result is
\begin{equation}\label{aa36}
\begin{split}
\bigg(\frac{\partial\textbf{H}_{an}}{\partial \textbf{p}}\bigg)_W=&\frac{m}{2^{m-1}}\bigg[\binom{m-1}{m-1}\textbf{p}^{m-1}\textbf{q}+\binom{m-1}{m-2}\textbf{p}^{m-2}\textbf{q}\textbf{p}+...+\binom{m-1}{0}\textbf{q}\textbf{p}^{m-1}\bigg]
\\&+i\hbar m(m-1)\textbf{p}^{m-2}.
\end{split}
\end{equation}
This is equivalent to
\begin{equation}\label{aa37}
\begin{split}
\bigg(\frac{\partial\textbf{H}_{an}}{\partial \textbf{p}}\bigg)_W=m\textbf{T}_{m-1,1}+i\hbar m(m-1)\textbf{p}^{m-2}.
\end{split}
\end{equation}
In the symmetric basis modification, the derivatives are
\begin{equation}\label{aa38}
\begin{split}
\bigg(\frac{\partial\textbf{H}}{\partial \textbf{p}}\bigg)_S&=\frac{m}{2}\bigg[\textbf{p}^{m-1}\textbf{q}+\textbf{q}\textbf{p}^{m-1}\bigg]
\\& =m\textbf{S}_{m-1,1}
\end{split}
\end{equation}
and
\begin{equation}\label{aa39}
\begin{split}
\bigg(\frac{\partial\textbf{H}_{an}}{\partial \textbf{p}}\bigg)_S&=\frac{m}{2}\bigg[\textbf{p}^{m-1}\textbf{q}+\textbf{q}\textbf{p}^{m-1}\bigg]+i\hbar m(m-1)\textbf{p}^{m-2}
\\& =m\textbf{S}_{m-1,1}+i\hbar m(m-1)\textbf{p}^{m-2}.
\end{split}
\end{equation}
Comparing \eqref{aa35} and \eqref{aa37}, then \eqref{aa38} and \eqref{aa39}, the derivatives have different values depending if the Hamiltonian is written as $\textbf{H}$ or $\textbf{H}_{an}$. An extra term $i\hbar m(m-1)\textbf{p}^{m-2}$ arises for all the three cases, posing an issue on the validity of the differential quotient of first type.

This form of the Hamiltonian is different from the ones discussed in the previous sections, where it is written as
\begin{equation}\label{aa40}
\textbf{A}_{m,n}:=\sum_{j=0}^{n}a_j^{(n)}\textbf{q}^j\textbf{p}^m\textbf{q}^{n-j}\quad\quad\quad\text{or}\quad\quad\quad\ \textbf{A}_{m,n}:=\sum_{j=0}^mb_j^{(m)}\textbf{p}^j\textbf{q}^n\textbf{p}^{m-j}.
\end{equation}
Ordering rules dictate the coefficients $a_j^n$ and $b_j^m$, along with the hermiticity of the operator. Legitimate Hamiltonians follow this form and were proven to be correct in their own respective basis, in the context of forming the quantum equations of motion as shown in Section \ref{2.4}. Writing operators in this form guarantee a consistent result when using the differential quotient of first type.

\end{document}